\documentclass[preprint,authoryear]{elsarticle}

\usepackage{lineno}
\usepackage[bookmarks=false]{hyperref}
\usepackage{amssymb}
\usepackage{amsthm}

\usepackage{amsmath}
\usepackage{array}
\usepackage{ebproof}
\usepackage{multirow}

\usepackage{textcomp}
\usepackage{listings}
\setcitestyle{authoryear,open={((},close={))}}

\usepackage{chngcntr}
\usepackage{apptools}
\AtAppendix{\counterwithin{lemma}{section}}

\journal{Journal of \LaTeX\ Templates}

\setcitestyle{authoryear,open={(},close={)}}
\newcommand{\rpc}{$\lambda_{rpc}$}
\newcommand{\polyrpc}{$\lambda_{rpc}^{\forall}$}

\newcommand{\polycs}{$\lambda_{cs}^{\forall}$}
\newcommand{\stateenccs}{$\lambda_{cs}^{enc}$}
\newcommand{\statefulcs}{$\lambda_{cs}^{state}$}

\newcommand{\client}{\textbf{c}}
\newcommand{\server}{\textbf{s}}

\newcommand{\evalRPC}[3]{#1\Downarrow_{#2}#3}

\newcommand{\lamL}[3]{\lambda^{#1}#2.#3}
 
\newcommand{\subst}[2]{\{#1/#2\}}

\newcommand{\textsfReq}{\textsf{req}}
\newcommand{\req}[2]{\textsfReq(#1,#2)}

\newcommand{\textsfCall}{\textsf{call}}
\newcommand{\call}[2]{\textsfCall(#1,#2)}

\newcommand{\textsfGen}{\textsf{gen}}
\newcommand{\gen}[3]{\textsfGen(#1,#2,#3)}

\newcommand{\funL}[1]{\xrightarrow{#1}}    
  
\newcommand{\tyenv}{\Gamma}     
\newcommand{\tyenvExt}[2]{\Gamma,#1:#2}
\newcommand{\tyenvExtWith}[1]{\Gamma,#1}

\newcommand{\typing}[4]{#1\vdash_{#2} #3 : #4}

\newcommand{\Loc}{Loc}

\newcommand{\mono}[1]{[\![#1]\!]}

\newcommand{\html}[2]{\textlangle{#1}\textrangle{#2}\textlangle/{#1}\textrangle}
\newcommand{\ul}[1]{\html{ul}{#1}}
\newcommand{\li}[1]{\html{li}{#1}}

\renewcommand{\cite}[1]{\citep{#1}}

\newtheorem{lemma}{Lemma}[section]
\newtheorem{theorem}{Theorem}[section]

\begin{document}

\begin{frontmatter}

\title{A  Polymorphic RPC Calculus}
% \title{Elsevier \LaTeX\ template\tnoteref{mytitlenote}}
%\tnotetext[mytitlenote]{Fully documented templates are available in the elsarticle package on \href{http://www.ctan.org/tex-archive/macros/latex/contrib/elsarticle}{CTAN}.}

\author[jnu]{Kwanghoon Choi \fnref{fn1}}

\author[edinburgh,turinginstitute]{James Cheney \fnref{fn2}}

\author[edinburgh]{Simon Fowler}

\author[edinburgh,imperialcollege]{Sam Lindley \fnref{fn3}}

\address[jnu]{Chonnam National University, Gwangju, Republic of Korea}

\address[edinburgh]{The University of Edinburgh, Scotland, UK}

\address[turinginstitute]{The Alan Turing Institute, London, UK}

\address[imperialcollege]{Imperial College, London, UK}

\fntext[fn1]{This research was supported by the National Research Foundation of Korea grants
 from MSIP (No. 2017R1A2B4005138) and MoE (No. 2019R1I1A3A01058608).}
\fntext[fn2]{This work was supported by ERC Consolidator Grant Skye (grant number 682315).
}
\fntext[fn3]{This work was supported by EPSRC Programme Grant ``From Data Types to
Session Types---A Basis for Concurrency and Distribution'' (EP/K034413/1).}

%% Group authors per affiliation:
%\author{Elsevier\fnref{myfootnote}}
%\address{Radarweg 29, Amsterdam}
%\fntext[myfootnote]{Since 1880.}

%% or include affiliations in footnotes:
%\author[mymainaddress,mysecondaryaddress]{Elsevier Inc}
%\ead[url]{www.elsevier.com}

%\author[mysecondaryaddress]{Global Customer Service\corref{mycorrespondingauthor}}
%\cortext[mycorrespondingauthor]{Corresponding author}
%\ead{support@elsevier.com}

%\address[mymainaddress]{1600 John F Kennedy Boulevard, Philadelphia}
%\address[mysecondaryaddress]{360 Park Avenue South, New York}

\begin{abstract}
	The RPC calculus is a simple semantic foundation for multi-tier programming languages such as Links in which located functions can be written for the client-server model. Subsequently, the typed RPC calculus is designed to capture the location information of functions by types and to drive location type-directed slicing compilations.
	However, the use of locations is currently limited to monomorphic ones, which is one of the gaps to overcome to put into practice the theory of RPC calculi for client-server model. 
	
	This paper proposes a polymorphic RPC calculus to allow programmers to write succinct multi-tier programs using polymorphic location constructs.
	Then the polymorphic multi-tier programs can be automatically translated into programs only containing location constants amenable to the existing slicing compilation methods.  
	We formulate a type system for the polymorphic RPC calculus, and prove its type soundness. Also, we design a monomorphization translation together with proofs on its type and semantic correctness for the translation. 
\end{abstract}


\end{frontmatter}

%\linenumbers

%% main text
\section{Introduction}
\label{sec:introduction}
	Multi-tier programming languages for the client-server model are designed to address the client-server dichotomy. 
	For example, a web system basically consists of a web server that accesses databases and a web client that provides user interfaces, and they are connected by a network. Programmers have to develop two individual programs separately for the two machines, which increases the programmer's burden. 
	Once a program is developed, programmers need to test the two programs together, which is more complex than with one program on a single machine. 
	After that, the integrity between the two programs should be properly maintained when each of them evolves. Further, some tasks cross the boundary of two computers, and so the separate development increases coupling between the client and server modules for the tasks.

%	Multi-tier programming languages for the client-server model are designed to address the client-server dichotomy. For example, in web programs, one develops a JavaScript-based client program and a Java-based server program in two different programming languages. The use of two different programming languages may demand more effort in development and also in testing. Moreover, modules may cross the boundary of two computers. This is what we call the problem of client-server dichotomy.
	
	Multi-tier programming attempts to solve this problem by allowing programmers to write a unified program for client and server expressions together in a single programming language, and by providing a slicing compilation method that can slice the unified program into separate client and server programs automatically.
	
	The untyped RPC calculus \cite{Cooper:2009:RC:1599410.1599439} is the semantic foundation for the RPC (remote procedure call) feature of Links \cite{Cooper:2006:LWP:1777707.1777724}, which is a functional programming language for multi-tier web programming. Here is an example excerpted from the RPC calculus paper, rewritten in Links as: 

\ \ 
\begin{minipage}{0.9\textwidth}
\begin{lstlisting}[escapechar=\%,language=lisp]
fun main() client { %\underline{authenticate}% () }

fun authenticate() server { 
   var creds = %\underline{getCredentials}%( "Enter name:passwd > " )
   if ( creds == "ezra:opensesame" ) { 
     "The secret document" 
   } else { "Access denied" }
}

fun getCredentials(prompt) client { (print(prompt);  read) }
\end{lstlisting}
\end{minipage}

	There are location attributes {\texttt client} and {\texttt server} that indicate where the associated functions should run. The program begins with $main$, which is a client function. It invokes a server function $authenticate$. Subsequently, in the body of the server function, a client function $getCredentials$ is invoked. These are two examples of remote procedure call (RPC), one from the client to the server and the other from the server to the client. 
	The RPC calculus thus expresses remote procedure calls as plain lambda applications. It also expresses local procedure calls such as $print (prompt)$ as the same syntax of lambda application. As a result, it may take some time to see if a given lambda application is a remote procedure call, particularly when higher-order functions are extensively used. The two remote procedure calls in the example are underlined. 
	
		Then a slicing compilation in the untyped RPC calculus slices a unified program such as the one above into a client program and a server program where each separate program contains only functions that must run at their own location and the remote procedure calls are compiled with some communication primitives. 
	
%	There are other multi-tier programming languages as well: ML5 \cite{MurphyVII:2004:SML:1018438.1021865,Murphy:2007:TDP:1793574.1793585,Murphy:2008:MTM:1467784}, Ur/Web \cite{Chlipala2015}, Hop \cite{Serrano2006,Serrano:2012:MPH:2240236.2240253,Serrano:2016:GH:3022670.2951916}, Eliom \cite{Radanne2017,Radanne:2018:TWP:3184558.3185953}, Multi-tier calculus \cite{Neubauer:2005:SPM:1040305.1040324}, ScalaLoci \cite{Weisenburger:2018:DSD:3288538.3276499}, and so on. But they have supported only asymmetric communication, or their semantics is more relevant to the peer-to-peer model, rather than the client-server model. Only the RPC calculus supports symmetric communication, which is convenient for programmers, and its semantics can be immediately implemented on the client-server model, which is suitable for compiler writers.
	
%	The original untyped RPC calculus poses two limitations. One limitation is that it only supports the  stateless server implementation. This is suitable for the scalable web model, but sometimes the  stateful server implementation is also required. The other one is that the compilation rules for the calculus are complicated because of the lack of location information in  application expressions. 
	
	The typed RPC calculus \cite{choijfp2019} is an extension with location types to specify where functions must run. For example, \texttt{authenticate} has type $Unit \funL{\server} String$ while \texttt{getCredentials} has type $String \funL{\client} String$ where $\server$ denotes the server and $\client$ does the client. It is equipped with a type system that can account for remote procedure calls at the type level as is done in the previous example by underlining them.
	Also, it provides a type-directed slicing compilation method simpler than the untyped slicing compilation method. Thanks to the simplicity, the method offers a spectrum of slicing compilations: one for the  stateless server where no states are maintained in the server, which is good for scalability, and the other for the stateful server where the server maintains all states during multiple interactions with the client. The method even suggests an idea of how to mix the two styles. The details for these slicing compilations are in  \cite{choijfp2019}.

	Note that in Links, location attributes are hints that are used at run-time rather than part of the type as in the typed RPC calculus.
	
	In spite of the advancement, there are still gaps in putting into practice the theory of RPC calculi for client-server model. The typed RPC calculus is good for writing functions with specific locations such as web page modification and database accesses. But it is bad for writing location neutral functions such as list utilities and primitive type functions because programmers write them twice, one for the client and the other for the server. Instead, the calculus should allow programmers to write location neutral functions only once for the two locations. Then the compiler should be extended to translate them into location-specific versions, for example, one written in JavaScript for the client and the other written in OCaml for the server, automatically. It is much like the convenience of polymorphically typed functions that are written once but can be applied multiple times over different instantiated types. 
	%Second, there must be a location inference algorithm, which is essential for the typing derivation directed compilations. In the RPC calculus, location information is attached only on lambda abstractions, but it is required to be available on lambda applications as well. It is reasonable for programmers to write location information once on lambda abstraction, but it would be very cumbersome to force them to demand them to do it repeatedly on lambda application. % The previous research did not present any algorithm for that at all.
	
	An introduction of polymorphic locations to multi-tier programming languages including the RPC calculi poses a technical problem. In the RPC calculi, programmers write remote procedure calls in the same syntax as for local procedure calls. They provide no RPC keyword. This is believed to be a good design for programmers because this abstracts out a difference in terms of the use of client-server communication. For implementation, however, we need to distinguish  between these two kinds of procedure calls, one implemented by a jump instruction and the other by a RPC library, exposing the difference explicitly in terms. The typed RPC calculus does this by location types as, for example, every application is a remote procedure call if the location of the application is different from the location of a function to invoke. On introducing polymorphic locations, we will have {\it location variables} where such location information is uncertain at compile-time.  

	 In this paper, we propose a polymorphic RPC calculus, which is an extension of the typed RPC calculus with polymorphic locations. %, and we design a simple location inference algorithm for the polymorphic RPC calculus. 
	 A key idea behind the polymorphic RPC calculus is to introduce location variables on lambda abstractions, to abstract locations by {\it location abstraction}, $\Lambda l.M$, and to instantiate them by {\it location application}, $M[Loc]$ for some location $Loc$. For example, the location-neutral \texttt{map} function could be written in  Links extended with the feature of polymorphic locations as 
`fun map(f, xs) l \{ the body of map \}' where the location attribute is replaced with a location variable $l$. In other words,

	 \[
	 	map = \Lambda l. \lamL{l}{f}{ \lamL{l}{xs}{\ \cdots \textit{the body of map} \cdots }}
	\] 
	that has type $\forall l. (A \funL{l} B) \funL{l} ([A] \funL{l} [B])$. The type of $f$ is $A \funL{l} B$, the type of $xs$ is a list type whose elements have type $A$, i.e., $[A]$, and the ultimate return type is $[B]$. The map function should run at location $l$, which is specified by location application. To run it as a client function, we use a location application $map[\client]$, which becomes $\lamL{\client}{f}{ \lamL{\client}{xs}{\ \cdots}}$ of type $ (A \funL{\client} B) \funL{\client} ([A] \funL{\client} [B])$ by replacing all occurrences of $l$ with $\client$. To run it as a server function, $map[\server]$ will be used. 
	Thus, every polymorphic $\lambda$-abstraction can be regarded as a location-neutral one, and a choice of a location specific $\lambda$-abstraction is done by a location application with the location.  

	For implementation of the location polymorphism, we design a method to translate away both the location abstraction construct and the location application construct from polymorphic RPC terms. It then becomes possible to make use of the existing slicing compilation methods for the typed RPC calculus \cite{choijfp2019}. Combining the translation and each of the existing slicing compilations, we can obtain two new slicing compilation methods for the polymorphic RPC calculus. Figure \ref{fig:overview} shows an overview of the polymorphic RPC calculus. 
	
\begin{figure}[t]
\centering
\includegraphics[width=\textwidth]{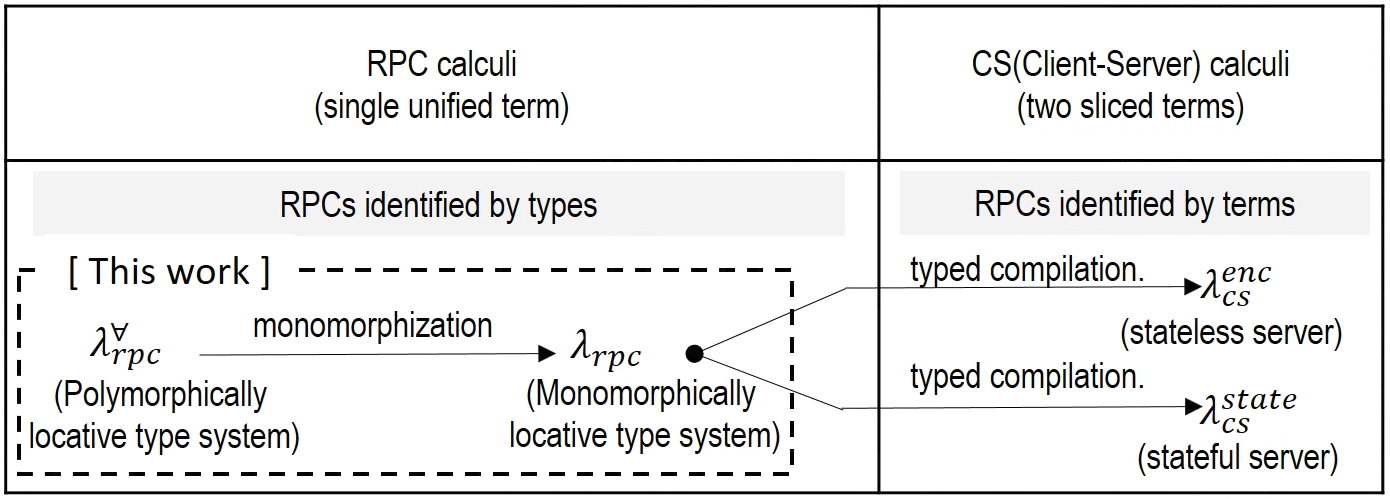}
\caption{Overview of a polymorphic RPC calculus}
\label{fig:overview}
\end{figure}	
	
	The contributions of this paper are as follows:
\begin{itemize}
	\item We propose a new polymorphic RPC calculus with the notion of polymorphic location, and prove its type soundness property. 
	\item We design a monomorphization translation for the polymorphic RPC calculus, and prove the type correctness and the semantic correctness of the translation.
\end{itemize}
	 
	 The roadmap of this paper is this. Section \ref{sec:rpc} reviews the typed RPC calculus as a background. In Section \ref{sec:polyrpc}, we propose a polymorphic RPC calculus and a monomorphization translation, and prove a few important properties. After we discuss related work in Section \ref{sec:relatedwork}, we conclude in Section \ref{sec:conclusion}.

\section{Background: The typed RPC calculus}
\label{sec:rpc}	

	In this section, we review the typed RPC calculus. Figure \ref{fig:rpc} shows a typed RPC calculus {\rpc}. It is a call-by-value $\lambda$-calculus with location annotations on $\lambda$-abstractions.  The location annotations tell where the $\lambda$-abstractions must execute. The client-server model is assumed in the design of the calculus, and so the location annotations are either $\client$ denoting client or $\server$ denoting server. The syntax of the typed RPC calculus is thus defined as shown in the figure. 

\begin{figure}[t]
\begin{tabular}{ l  l  l  c  l    c  l  c  l  c  l }
\multicolumn{11}{l}{\textbf{Syntax}} \\
  \ \ \ \ \
            & Location 		& $a,b$ 		& $::=$ 	& $\client$	& $|$		& $ \server$		 \\
            & Term 	& $L,M,N$	& $::=$ 	& $V$			& $|$		& $L \ M$		& 	$|$		&	$M[A]$ \\
            & Value    			& $V,W$		& $::=$ 	& $x$			& $|$		& $ \lambda^{a} x.M$ &		$|$		& $\Lambda\alpha.V$  
\\[0.1cm]
\multicolumn{11}{l}{\textbf{Semantics}} \\
 %&
  \multicolumn{11}{c}{
  \ \ 
  	\mbox{
       \begin{prooftree}
       		\Infer0[(Abs)]{ \evalRPC{\lamL{b}{x}{M}}{a}{\lamL{b}{x}{M} }}
	\end{prooftree}
	\ \ \ 
	\begin{prooftree}
  		\Hypo{ \evalRPC{L}{a}{\lamL{b}{x}{N}} }
  		\Hypo{ \evalRPC{M}{a}{W} }
  		\Hypo{ \evalRPC{N\subst{W}{x}}{b}{V}  }
  		\Infer3[(App)]{ \evalRPC{L \ M}{a}{V}  }
	\end{prooftree} 
  	}
    }
\\[0.5cm]
 %&
  \multicolumn{11}{c}{
  \ \ \ 
  	\mbox{
       \begin{prooftree}
       		\Infer0[(Tabs)]{ \evalRPC{\Lambda\alpha.V}{a}{\Lambda\alpha.V }}
	\end{prooftree}
	\ \ \ \ \ 
	\begin{prooftree}
  		\Hypo{ \evalRPC{M}{a}{\Lambda\alpha.V} }
  		\Infer1[(Tapp)]{ \evalRPC{M[B]}{a}{V\subst{B}{\alpha}}  }
	\end{prooftree} 
  	}
    } 
\end{tabular}
\caption{A typed RPC calculus \rpc{}}
\label{fig:rpc}
\end{figure}

	The semantics of {\rpc} is defined in a big-step operational semantics with evaluation judgment, $\evalRPC{M}{a}{V}$ that denotes the evaluation of a term $M$ at location $a$ resulting in a value $V$. (Abs) straightforwardly defines an evaluation relation between a location annotated $\lambda$-abstraction and itself. (App) is more interesting for an application $L \ M$ at location $a$: it performs $\beta$-reduction at location $b$, which a $\lambda$-abstraction from $L$ has as an annotation, with a value $W$ from $M$, and it continues to evaluate the $\beta$-reduced term at the location. Here, $N\subst{W}{x}$ is a substitution of $W$ for $x$ in $N$. 
	
	Note that in (App), $L \ M$ is a remote procedure call whenever the caller location $a$ is different from the callee location $b$. Otherwise, it is a local procedure call. When $a$ is client and $b$ is server, a server function is invoked from the client, and vice versa. The typed RPC calculus is a simple semantic foundation because it uses the same syntax of $\lambda$-application and the same evaluation rule (App) both for remote procedure calls and local ones. But every remote procedure call  must be implemented differently from local procedure calls since it involves communication between client and server. 
	
	The typed RPC calculus in the figure is in fact an extension with polymorphic types of  the original typed one \cite{choijfp2019}, but is still limited to monomorphic locations in the same way. So, the typed RPC calculus now has type abstraction $\Lambda\alpha.V$ and type application $M[A]$ where $\alpha$ is a type variable and $A$ is a type, which will be defined soon. (Tabs) and (Tapp) are quite the standard definitions for evaluation of type abstraction and type application. Note $V\subst{B}{\alpha}$ is a substitution of type $B$ for type variable $\alpha$ in $V$. 

	A type system for the typed RPC calculus in Figure \ref{fig:rpctysystem} basically comes from the one in  \cite{choijfp2019} that accounts for remote procedure calls in the type level. It extends the original type system with polymorphic types by having two standard typing rules for type abstraction and type application, (T-Tabs) and (T-Tapp). $A\subst{B}{\alpha}$ is a substitution of $B$ for each occurrence of $\alpha$ in $A$.
	Accordingly, every typing environment $\tyenv$ now has type variables as well as associations of variables and types in general as $\{\alpha_1,\cdots,\alpha_k, x_1:A_1, \cdots, x_m:A_m\}$. 
	%The typed RPC calculus provides a type system that accounts for remote procedure calls in type level, as shown in Figure \ref{fig:rpctysystem}. 

	The type system has two features related to location. First, location annotations are introduced to function types as $A \funL{a} B$. Every $\lambda$-abstraction that must run at location $a$ gets this function type. For example, $(\lamL{ \server }{f}{ \  (f \ M)  }) \ \ (\lamL{\client}{y}{  \cdots   })$ is well-typed when $f$ is of type $A\funL{\client} B$ for some types $A$ and $B$.  However, $(\lamL{\client}{f}{ \ \cdots (\lamL{\client}{h}{\ \cdots}) \ f \ \ \cdots \ (\lamL{\client}{g}{\ \cdots}) \ f } \cdots )$ is ill-typed when $h$ is of $A\funL{\client} B$ and $g$ is of $A\funL{\server}B$ since the type language of the typed RPC calculus  only allows either $\client$ or $\server$, not both of them on a function type. In this respect, this typed RPC calculus (or the original one \cite{choijfp2019}) is monomorphic in terms of specifying location of evaluation. 
	
	Second, location annotations are also attached on typing judgments as $\typing{\tyenv}{a}{M}{A}$ saying a term $M$ at location $a$ has type $A$ under a type environment $\tyenv$. 
	(T-Var) is defined as usual. 
	(T-Abs) assigns $\lambda$-abstraction a function type with the same location as its annotation. Note that a location on the typing judgment in the conclusion changes to the annotated location in the premise for the body of  $\lambda$-abstraction.
	
	Combining these two features, (T-App) is designed to be a refinement of the conventional $\lambda$-application typing with respect to the combinations of location $a$ (where to evaluate the application) and location $b$ (where to evaluate the function). When $a$ is different from $b$, $L \ M$ is statically found to be a remote procedure call: if $a=\client$ and $b=\server$, it is to invoke a server function from the client, and if $a=\server$ and $b=\client$, it is to invoke a client function from the server. Otherwise, one can statically decide that it is a local procedure call. 
	
	\begin{figure}[t]
\begin{tabular}{l l l l l l l l l l l l l }
\multicolumn{13}{l}{\textbf{Types}} \\
   \ \ \ \ \ & Type & $A,B,C$ & $::=$ & $base$ & $ | $ & $ A\funL{a}A$ & $|$ & $\alpha$ & $|$ & $\forall\alpha.A$ \\[0.1cm]
\multicolumn{13}{l}{\textbf{Typing Rules}} \\
 & 
 \multicolumn{12}{c}{  
  \mbox{
  	\begin{prooftree}
		\Hypo{  \tyenv(x)=A }
		\Infer[left label=(T-Var)]1{ \typing{\tyenv}{a}{x}{A} }
	\end{prooftree}
		\ \ \ 
	\begin{prooftree}
		\Hypo{ \typing{\tyenvExt{x}{A}}{b}{M}{B} }
		\Infer[left label=(T-Abs)]1{ \typing{\tyenv}{a}{\lamL{b}{x}{M}}{A\funL{b}B} } 
	\end{prooftree}
%		\ \ \ 
%	\begin{prooftree}
%		\Hypo{  \typing{\tyenv}{a}{L}{A\funL{b}B } }
%		\Hypo{  \typing{\tyenv}{a}{M}{A} }
%		\Infer[left label=(T-App)]2{ \typing{\tyenv}{a}{L \ M}{B}   }
%	\end{prooftree}
    } 
    }
\\[0.5cm] 
& 
 \multicolumn{12}{c}{  
 \mbox{
	\begin{prooftree}
		\Hypo{  \typing{\tyenv}{a}{L}{A\funL{b}B } }
		\Hypo{  \typing{\tyenv}{a}{M}{A} }
		\Infer[left label=(T-App)]2{ \typing{\tyenv}{a}{L \ M}{B}   }
	\end{prooftree}
	}
    }
\\[0.5cm]
& 
 \multicolumn{12}{c}{  
 \mbox{
	\begin{prooftree}
		\Hypo{  \typing{\tyenv,\alpha}{a}{V}{A} }
		\Infer[left label=(T-Tabs)]1{ \typing{\tyenv}{a}{\Lambda\alpha.V}{\forall\alpha.A}   }
	\end{prooftree}
	\ \ \ \ \
	\begin{prooftree}
		\Hypo{  \typing{\tyenv}{a}{M}{\forall\alpha.A} }
		\Infer[left label=(T-Tapp)]1{ \typing{\tyenv}{a}{M[B]}{A\subst{B}{\alpha}}   }
	\end{prooftree}
	}
    }
\end{tabular}
\caption{A type system for the  typed RPC calculus}
\label{fig:rpctysystem}
\end{figure}

	The type soundness theorem for the typed RPC calculus   \cite{choijfp2019} guarantees that every remote procedure call thus identified statically will never change to a local procedure call under evaluation. The two slicing compilations for the typed RPC calculus depend on this capability of the type system as an analysis on dynamic communication patterns. 
	
	As shown in Figure \ref{fig:overview}, the typed RPC calculus offers slicing compilation methods to slice a unified RPC program in {\rpc} into a client program and a server program in the client-server (CS) calculi {\stateenccs} or {\statefulcs}, automatically.
	Each separate program will contain only functions that must run at one's own location, and the remote procedure calls in the unified program will be compiled with some communication primitives in the client-server programs. 
	%There are two slicing compilation methods, one for stateless server style and the other for stateful server \cite{choijfp2019}. 

	The slicing methods are type-directed compilations
	%Regardless of which slicing method is chosen, they are typed compilation methods 
	where every input RPC program is type-checked to produce a typing derivation for it and then  sliced programs are generated. In a unified RPC program, there is no particular syntax to specify remote procedure calls but only location types can be used to identify them. In separate client and server programs, we introduce explicit constructs for remote procedure calls like this. A construct $\req{V}{W}$ is used to invoke a server function $V$ with an argument $W$ from the client, and $\call{V}{W}$ is another remote procedure call construct in the reverse direction. In separate programs, we use $V(W)$ only for local procedure calls, differently from what we do in  a unified program. 

	Now we are ready to present a key idea of the type-directed slicing compilations. Given a well-typed unified RPC program, each use of (T-App) in the typing derivation is compiled differently depending on the combination of the two locations on where to evaluate the application ($a$) and where to evaluate the function ($b$), as is explained above.   When $a$ is the same as $b$, the slicing compilation methods generate a normal application term, say, $V(W)$ where $V$ is a local function and $W$ is an argument. When $a=\client$ and $b=\server$, the slicing compilation methods generate $\req{V}{W}$ where $V$ is a server function. This term is implemented as sending $V$ and $W$ to the server to apply the function to the argument there and  receiving either the application result or a new server-side call to invoke a client function. When $a=\server$ and $b=\client$, the slicing compilation methods generate $\call{V}{W}$ where $V$ is a client function. This term is implemented in a similar way for $\req{-}{-}$ but in the reverse direction from the server to the client. In sliced programs, only applications of the form $V(W)$ do not involve communication at all. For details, the reader can refer to the formal semantics for the client-server calculi ({\stateenccs} and {\statefulcs}) and the two slicing compilation rules in \cite{choijfp2019}.
	
	Although the original typed RPC calculus does not consider type polymorphism for a simple presentation, we see no problem in applying the two slicing compilations to this typed RPC calculus in the presence of type polymorphism assuming the target language also has type polymorphism.

\section{A typed RPC calculus extended with polymorphic locations}
\label{sec:polyrpc}

	In this section, we firstly extend the typed RPC calculus with the notion of polymorphic location to write polymorphic functions seamlessly with monomorphic functions, which is convenient for programmers.  We call it a {\it polymorphic RPC calculus}, {\polyrpc}. Secondly, we design a translation of the polymorphic RPC calculus into the typed RPC calculus. Then we are able to make use of the two existing slicing compilation methods even for the polymorphic RPC calculus. 

\subsection{A polymorphic RPC calculus}

	An important feature of the polymorphic RPC calculus is the notion of location variable $l$ for which we can substitute a  location (constant) $a$. Accordingly, a new syntactic object $\Loc$ is introduced to be either a location constant or a location variable. 

	Every $\lambda$-abstraction $\lamL{\Loc}{x}{M}$ now has an annotation of $\Loc$ instead of $a$. By substituting a location $b$ for a location variable annotation,  $(\lamL{l}{x}{M})\subst{b}{l}$ becomes a monomorphic $\lambda$-abstraction $\lamL{b}{x}{(M\subst{b}{l})}$. This location variable is abstracted by the location abstraction construct $\Lambda l.V$, and it is instantiated by the location application construct $M[\Loc]$. Except for these three constructs, the other syntax is the same as that in the typed RPC calculus. Figure \ref{fig:polyrpc} summarizes the syntax of the polymorphic RPC calculus. 
	
\begin{figure}[t]
\begin{tabular}{ l  l  l  c  l    c  l  c  l  c  l c l }
\multicolumn{13}{l}{\textbf{Syntax}} \\
  \ \ \ 
            & Location 		& $a,b$ 		& $::=$ 	& $\client$	& $|$		& $ \server$		 \\
            &                   		& $\Loc$   	& $::=$  	& $a$              	& $|$          & $l$   \\
            & Term 			& $L,M,N$	& $::=$ 	& $V$			& $|$		& $L \ M$	 					& 	$|$		&	$M[A]$ & $|$	& $M[\Loc]$	 \\
            & Value    			& $V,W$		& $::=$ 	& $x$			& $|$		& $ \lamL{\Loc}{x}{M}$	&		$|$		& $\Lambda\alpha.V$ 	& $|$	& $\Lambda l.V$
\\[0.1cm]
\multicolumn{13}{l}{\textbf{Semantics}} \\
 % & 
 \multicolumn{13}{c}{
  	\mbox{
       \begin{prooftree}
       		\Infer0[(Abs)]{ \evalRPC{\lamL{b}{x}{M}}{a}{\lamL{b}{x}{M} }}
	\end{prooftree}
	\ \ \ 
	\begin{prooftree}
  		\Hypo{ \evalRPC{L}{a}{\lamL{b}{x}{N}} }
  		\Hypo{ \evalRPC{M}{a}{W} }
  		\Hypo{ \evalRPC{N\subst{W}{x}}{b}{V}  }
  		\Infer3[(App)]{ \evalRPC{L \ M}{a}{V}  }
	\end{prooftree} 
  	}
    }   
\\[0.5cm]
 %&
  \multicolumn{13}{c}{
  \ \ \ 
  	\mbox{
       \begin{prooftree}
       		\Infer0[(Tabs)]{ \evalRPC{\Lambda\alpha.V}{a}{\Lambda\alpha.V }}
	\end{prooftree}
	\ \ \ \ \ 
	\begin{prooftree}
  		\Hypo{ \evalRPC{M}{a}{\Lambda\alpha.V} }
  		\Infer1[(Tapp)]{ \evalRPC{M[B]}{a}{V\subst{B}{\alpha}}  }
	\end{prooftree} 
  	}
    } 
\\[0.5cm]
 %& 
 \multicolumn{13}{c}{
  	\mbox{
	\begin{prooftree}
       		\Infer0[(Labs)]{ \evalRPC{\Lambda l.V}{a}{\Lambda l.V} }
	\end{prooftree}
	\ \ \ \ \ \ \ 
	\begin{prooftree}
  		\Hypo{ \evalRPC{M}{a}{\Lambda l.{V}} }
    		\Infer1[(Lapp)]{ \evalRPC{M[b]}{a}{V\subst{b}{l}}  }
	\end{prooftree} 
  	}
    }     
\end{tabular}
\caption{The polymorphic  RPC calculus \polyrpc}
\label{fig:polyrpc}
\end{figure}

	The semantics for the polymorphic RPC calculus is  shown in Figure \ref{fig:polyrpc}. Every location abstraction is regarded as a value; it evaluates to itself by (Labs). Every location application $M[\Loc]$ firstly evaluates to a location  abstraction, and then $\Loc$ is substituted for the location variable in the body of the abstraction by (Lapp).
	
	A new form of substitution $M\subst{\Loc}{l}$ is defined as this. 
	\begin{eqnarray*}
	x\subst{\Loc}{l} & = & x \\
	(\lamL{\Loc}{x}{M})\subst{\Loc'}{l} & = & \lamL{\Loc\subst{\Loc'}{l}}{x}{M\subst{\Loc'}{l}} \\
	(\Lambda\alpha.M)\subst{\Loc}{l} & = & \Lambda\alpha. (M \subst{\Loc}{l}) \\
	(\Lambda l.V)\subst{\Loc}{l'}  & = & \left\{\mbox{\begin{tabular}{l l}
													$\Lambda l.V$ &  if  $l=l'$ \\
	                                                         				$\Lambda l.(V\subst{\Loc}{l'})$ & otherwise
												\end{tabular}
												
											} \right.
											 \\
	(L \ M)\subst{\Loc}{l} & = & (L \subst{\Loc}{l}) (M \subst{\Loc}{l}) \\	
	(M[A]) \subst{\Loc}{l} & = & (M \subst{\Loc}{l}) (A \subst{\Loc}{l}) \\
      (M[\Loc])\subst{\Loc'}{l} &=& (M\subst{\Loc'}{l})[\Loc\subst{\Loc'}{l}]
	\end{eqnarray*}
      where the definition of $\Loc\subst{\Loc'}{l}$ is
      	\begin{eqnarray*}
      \Loc\subst{\Loc'}{l} & = & \left\{\mbox{\begin{tabular}{l l}
													$\Loc$ & if $\Loc=a$ \\
	                                                         				$\Loc'$ & if $\Loc=l'$ and $l=l'$ \\
	                                                         				$\Loc$ & if $\Loc=l'$ and $l\not=l'$
												\end{tabular}
												
											} 
								\right.
	\end{eqnarray*}
	and the definition of $A \subst{\Loc}{l}$ will be presented below.

	For example, a polymorphic identity function $id$ is defined as  $\Lambda l.\lamL{l}{x}{x}$ so that  we can use  $id[\client]$ as $\lamL{\client}{x}{x}$,  and $id[\server]$ as $\lamL{\server}{x}{x}$. Thus, every polymorphic $\lambda$-abstraction can be regarded as a function usable both in the client and in the server, and a choice of a location specific $\lambda$-abstraction is done by a location application with the location. 
	
\begin{figure}[t]
\begin{tabular}{l l l l l l l l l l l l l }
\multicolumn{13}{l}{\textbf{Types}} \\
   \ \ \  & Type & $A,B,C$ & $::=$ & $base$ & $ | $ & $ A\funL{\Loc}B$  & $|$ & $\alpha$ & $|$ & $\forall\alpha.A$ & $ | $ & $ \forall l.A$ \\[0.1cm]
\multicolumn{13}{l}{\textbf{Typing Rules}} \\
 & 
 \multicolumn{12}{c}{  
  \mbox{
  	\begin{prooftree}
		\Hypo{  \tyenv(x)=A }
		\Infer[left label=(T-Var)]1{ \typing{\tyenv}{\Loc}{x}{A} }
	\end{prooftree}
	\ \ \ \ \ 
	\begin{prooftree}
		\Hypo{ \typing{\tyenvExt{x}{A}}{\Loc}{M}{B} }
		%\Hypo{ flv(\Loc) \subseteq \Delta  }
		\Infer[left label=(T-Abs)]1{ \typing{\tyenv}{\Loc'}{\lamL{\Loc}{x}{M}}{A\funL{\Loc}B} } 
	\end{prooftree}
    } 
    }
\\[0.5cm] 
& 
 \multicolumn{12}{c}{  
 \mbox{
	\begin{prooftree}
		\Hypo{  \typing{\tyenv}{\Loc}{L}{A\funL{\Loc'}B } }
		\Hypo{  \typing{\tyenv}{\Loc}{M}{A} }
		%\Hypo{  flv{A}\cup flv(\Loc') \subseteq \tyenv  }
		\Infer[left label=(T-App)]2{ \typing{\tyenv}{\Loc}{L \ M}{B}   }
	\end{prooftree}
	}
    }
\\[0.5cm]
& 
 \multicolumn{12}{c}{  
 \mbox{
	\begin{prooftree}
		\Hypo{  \typing{\tyenv,\alpha}{\Loc}{V}{A} }
		\Infer[left label=(T-Tabs)]1{ \typing{\tyenv}{\Loc}{\Lambda\alpha.V}{\forall\alpha.A}   }
	\end{prooftree}
	\ \ \ \ \
	\begin{prooftree}
		\Hypo{  \typing{\tyenv}{\Loc}{M}{\forall\alpha.A} }
		\Infer[left label=(T-Tapp)]1{ \typing{\tyenv}{\Loc}{M[B]}{A\subst{B}{\alpha}}   }
	\end{prooftree}
	}
    }
\\[0.5cm]
& 
\multicolumn{12}{c}{  
  \mbox{
	\begin{prooftree}
		\Hypo{ \typing{\tyenvExtWith{l}}{\Loc}{V}{A} }
		%\Hypo{ l \not\in flv(\varenv)\cup flv(\tyenv) \cup flv(\Loc) }
		\Infer[left label=(T-Labs)]1{ \typing{\tyenv}{\Loc}{\Lambda l.V}{\forall l.A }} 
	\end{prooftree}
	\ \ \ \ \ 
	\begin{prooftree}
		\Hypo{ \typing{\tyenv}{\Loc}{M}{\forall l.A } }
		% \Hypo{ flv(\Loc')\subseteq\tyenv  }
		\Infer[left label=(T-Lapp)]1{ \typing{\tyenv}{\Loc}{M[\Loc']}{A\subst{\Loc'}{l}}} 
	\end{prooftree}
    } 
    }
\end{tabular}
\caption{A type system for the polymorphic  RPC calculus}
\label{fig:polyrpctysystem}
\end{figure}

	Figure \ref{fig:polyrpctysystem} shows a type system for the polymorphic RPC calculus. The type language allows function types to  have the new syntactic object $\Loc$ as their annotation as $A \funL{\Loc} B$. Then every $\lambda$-abstraction at unknown location  gets assigned $A\funL{l} B$ using some location variable $l$. A universal quantifier over a location variable, $\forall l. A$, is also introduced to the type language  accordingly. 
	
	We extend a typing judgment $\typing{\tyenv}{\Loc}{M}{A}$ with two things. First, a location variable can be annotated on it. This extension is used to assert a typing relation in the body of  $\lamL{l}{x}{M}$. Second, typing environments now include location variables as well. 
	A general form of $\tyenv$ can now be written as $\{\alpha_1,\cdots,\alpha_k, l_1, \cdots,l_n,x_1:A_1, \cdots, x_m:A_m\}$. 
	This extension is used to keep track of a set of usable location variables. 
	The domain of environment, $dom(\tyenv)$, is defined as a union of type, location, and term variables as $\{ \alpha_1,\cdots,\alpha_k,l_1, \cdots,l_n, x_1, \cdots, x_m \}$, and the range, $rng(\tyenv)$, is $\{A_1,\cdots,A_m\}$. 
	
%	A location variable environment $\Delta$ is defined as a set of location variables $\{ l_1, \cdots, l_n \}$. 
	
	Recall that $(\lamL{c}{f}{ \ \cdots (\lamL{c}{h}{\ \cdots}) \ f \ \ \cdots \ (\lamL{c}{g}{\ \cdots}) \ f } \cdots )$ is an ill-typed term in the typed RPC calculus,   where $h$ is of type $A\funL{\client} B$ and $g$ is of type $A\funL{\server} B$. The polymorphic RPC calculus can make the term be well-typed by assigning $f$ a polymorphic type as $\forall l. A \funL{l} B$ and by slightly changing the first and second occurrences of $f$ into $f[\client]$ and $f[\server]$ respectively. 

	As a set of free variables is defined by the conventional definition of $fv(M)$, we define a set of free location variables over various kinds of objects in the form as $flv(-)$. 
	Two definitions for locations and types are as follows.  
	\begin{eqnarray*}
		flv(a) &=& \{ \} \\
		flv(l)  &=& \{ l\} \\
		\ \\
%	\end{eqnarray*}
%	For types, 
%	\begin{eqnarray*}
		flv(base) &=& \emptyset \\
		flv(\alpha) &=& \emptyset \\
		flv(A\funL{\Loc}B)) &=& flv(A)\cup flv(\Loc)\cup flv(B) \\
		flv(\forall \alpha.A) &=& flv(A) \\
		flv(\forall l.A) &=& flv(A) \backslash \{l\}
	\end{eqnarray*}
	
	For typing environments, it is a union of location variables there and free location variables occurring in types associated with variables as 
	\[ flv(\{ \alpha_1,\cdots,\alpha_k, l_1, \cdots,l_n,x_1:A_1, \cdots, x_m:A_m \}) = \{ l_1, \cdots, l_n \} \  \cup \bigcup_{1\leq i \leq m} flv(A_i)
	\]
	
	The five typing rules  (T-Var), (T-Abs), (T-App), (T-Tabs), and (T-Tapp) for the polymorphic RPC calculus generalize those for the typed RPC calculus by having $\Loc$ on function types and on typing judgments. Two new typing rules (T-Labs) and (T-Lapp) are similar to the typing rules for type abstraction and type application. (T-Labs) checks if the body of the location abstraction is typed with an extended typing environment with a fresh location variable. %its bound location variable does not appear in the type environment and in the contextual location. 
(T-Lapp) substitutes $\Loc'$ for all occurrences of a location variable $l$ on $\lambda$-abstractions in $M$. 
	
	The definition of $A\subst{\Loc}{l}$ is this.
	\begin{eqnarray*}
	base \subst{\Loc}{l}  & = & base \\
	(A \funL{\Loc} B) \subst{\Loc'}{l}  & = & A \subst{\Loc'}{l} \funL{ \Loc \subst{\Loc'}{l} } B \subst{\Loc'}{l} \\
	\alpha \subst{\Loc}{l} & = & \alpha \\ 
	(\forall\alpha.A) \subst{\Loc}{l}  & = & \forall\alpha. (A \subst{\Loc}{l} ) \\
	(\forall l.A)\subst{\Loc}{l'}  & = & \left\{\mbox{\begin{tabular}{l l}
													$\forall l.A$ &  if  $l=l'$ \\
	                                                         				$\forall l.(A\subst{\Loc}{l'})$ & otherwise
												\end{tabular}
												
											} \right.
	\end{eqnarray*}
	
	From now on, we will consider only well-formed typing judgments where there are no unbound variables, no unbound type variables, and no unbound free location variables. That is, given $\typing{\tyenv}{\Loc}{M}{A}$, we will safely assume three things. First, $fv(M) \subseteq dom(\tyenv)$. Second, $ \bigcup_{A_i \in rng(\tyenv)} ftv(A_i) \cup ftv(M) \cup ftv(A) \subseteq dom(\tyenv)$ where $ftv(A)$ or $ftv(M)$ are the sets of free type variables occurring in the type and the term respectively. They can be defined straightforwardly. Third, $\bigcup_{A_i \in rng(\tyenv)} flv(A_i) \cup flv(\Loc) \cup flv(M) \cup flv(A) \subseteq dom(\tyenv)$. 
	For example, in (T-App), location variables occurring in $A$ and $\Loc'$ are from any location abstractions enclosing the application term, in (T-Tabs), a bound type variable $\alpha$ does not occur as a free type variable in $\tyenv$, and 
	in (T-Labs), a bound  location  variable $l$ never occurs as a free location variable in $\tyenv$ and $\Loc$. 

%	Although the typing rules are a little loosely defined for the sake of simplicity, we safely assume that, in (T-App), location variables occurring in $A$ and $\Loc'$ must not be arbitrary, but they must be from any location abstractions enclosing the application term. We could make this precise by introducing a location variable environment. We would add a new bound location variable to it for each location abstraction, and we reference any location variables only from it. 
	
\subsubsection{Type soundness}
	
	Now we are ready to prove the type soundness of the type system for the polymorphic RPC calculus. Theorem \ref{thm:typesoundness} formulates this property. Its proof is done by induction on the height of evaluation derivations, and is available below. 
	The value substitution  lemma (\ref{lemma:valuesubst})  offers a typing hypothesis for the $\beta$-reduced term necessary for proving the case (App) of the theorem. 
	The type substitution lemma (\ref{lemma:typesubst}) proves a typing for a type substituted value from a type abstraction. 
	The value relocation lemma (\ref{lemma:valueloc}) is used to prove that the movement of a return value from a callee location to a caller location does not change its type.  
	The location substitution lemma (\ref{lemma:locsubst}) results in proving the case (Lapp) of the theorem by offering a typing hypothesis for the location-reduced term.  
	The proofs of all these lemmas are available in the appendix.
	
%%%
 \begin{lemma}[Value relocation]
 \label{lemma:valueloc} If \ $\typing{\tyenv}{\Loc}{V}{A}$ then $\typing{\tyenv}{\Loc'}{V}{A}$.
 \end{lemma}

 \begin{lemma}[Value substitution] 
 \label{lemma:valuesubst}
If \ $\typing{\tyenv}{\Loc}{\lamL{\Loc'}{x}{M}}{A\funL{\Loc'}B}$ and $\typing{\tyenv}{\Loc}{V}{A}$ then $\typing{\tyenv}{\Loc'}{M\subst{V}{x}}{B}$. 
 \end{lemma}
 
  \begin{lemma}[Type substitution] 
 \label{lemma:typesubst}
If \ $\typing{\tyenv}{\Loc}{\Lambda\alpha.V}{\forall\alpha.A}$ then $\typing{\tyenv}{\Loc}{V\subst{B}{\alpha}}{A\subst{B}{\alpha}}$. 
 \end{lemma}
 
 \begin{lemma}[Location substitution]
 \label{lemma:locsubst} If \ $\typing{\tyenv}{\Loc}{\Lambda l.V}{\forall l.A}$ then $\typing{\tyenv}{\Loc}{V\subst{\Loc'}{l}}{A\subst{\Loc'}{l}}$.
 \end{lemma}

\begin{theorem}[Type soundness for \polyrpc] For a  closed term $M$, if \ $\typing{\emptyset}{a}{M}{A}$ and $\evalRPC{M}{a}{V}$, then $\typing{\emptyset}{a}{V}{A}$.
\label{thm:typesoundness}
\end{theorem}
\begin{proof}
We prove this theorem by induction on the height of the evaluation derivation. Base cases use (Abs),  (Labs), and (Tabs) while inductive cases involve (App), (Lapp), and (Tapp). 

{\bf Case (Abs): } 
	$\evalRPC{\lamL{b}{x}{M_0}}{a}{\lamL{b}{x}{M_0}}$ where  $M=V=\lamL{b}{x}{M_0}$. Therefore, this case holds by the typing judgment hypothesis.
	
{\bf Case (Labs): } 
	$\evalRPC{\Lambda l.V_0}{a}{\Lambda l.V_0}$ where $M=V=\Lambda l.V_0$. Again, the typing judgment hypothesis proves this case.

{\bf Case (Tabs): }
	$\evalRPC{\Lambda\alpha.V_0}{a}{\Lambda\alpha.V_0}$ where $M=V=\Lambda\alpha.V_0$. Once again, the typing judgment hypothesis proves this case.
	
{\bf Case (App):  }
	$\evalRPC{L \ M_1}{a}{V}$, (1):$\evalRPC{L}{a}{\lamL{b}{x}{M_0}}$, (2):$\evalRPC{ M_1}{a}{W}$, and (3):$\evalRPC{M_0 \subst{W}{x}}{b}{V}$ where $M=L \ M_1$.  
	By $\typing{\emptyset}{a}{L \ M_1}{B}$, (4):$\typing{\emptyset}{a}{L}{A_1 \funL{\Loc'} B}$, and (5):$\typing{\emptyset}{a}{M_1}{A_1}$ where $B=A$. 

	By applying I.H. to (1) and (4), (6):$\typing{\emptyset}{a}{\lamL{b}{x}{M_0}}{A_1 \funL{\Loc'} B}$ where $\Loc'=b$.  
	By applying I.H. to (2) and (5), (7):$\typing{\emptyset}{a}{W}{A_1}$. 
	
	From (6) and (7), (8):$\typing{\emptyset}{b}{M_0 \subst{W}{x}}{B}$ is implied by the lemma (Value substitution).

	By applying I.H to (3) and (8), (9):$\typing{\emptyset}{b}{V}{B}$. By the lemma (Value relocation) with (9), $\typing{\emptyset}{a}{V}{B}$. Since $B=A$, this case is proved. 

{\bf Case (Lapp): } 
	$\evalRPC{L[b]}{a}{V_0\subst{b}{l}}$, (1):$\evalRPC{L}{a}{\Lambda l.V_0}$ where $M=L[b]$ and $V=V_0\subst{b}{l}$. 
	
	$\typing{\emptyset}{a}{L[b]}{A_0\subst{b}{l}}$, (2):$\typing{\emptyset}{a}{L}{\forall l. A_0}$ where $\tyenv=\emptyset$ and $A=A_0\subst{b}{l}$. 
	
	By applying I.H. to (1) and (2), (3):$\typing{\emptyset}{a}{\Lambda l.V_0}{\forall l.A_0}$. 

	The lemma (Location substitution) with (3) implies (4):$\typing{\emptyset}{a}{V_0\subst{b}{l}}{A_0\subst{b}{l}}$, which proves this case.
	
{\bf Case (Tapp):}
	$\evalRPC{M_1[B]}{a}{V_0\subst{B}{\alpha}}$, (1):$\evalRPC{M_1}{a}{\Lambda\alpha.V_0}$ where $M=M_1[B]$ and $V=V_0\subst{B}{\alpha}$. 
	
	By $\typing{\emptyset}{a}{M_1[B]}{A_0\subst{B}{\alpha}}$, (2):$\typing{\emptyset}{a}{M_1}{\forall\alpha.A_0}$ where $A=A_0\subst{B}{\alpha}$. 
	
	By applying I.H. to (1) and (2), (3):$\typing{\emptyset}{a}{\Lambda\alpha.V_0}{\forall\alpha.A_0}$. 
	
	By the lemma (Type substitution) with (3), $\typing{\emptyset}{a}{V_0\subst{B}{\alpha}}{A_0\subst{B}{\alpha}}$, which proves this case.

\end{proof}

\subsubsection{An example using polymorphic location}
\label{subsec:polyrpcexample}

	To explain how polymorphic location is useful in writing multi-tier programs, we will present an example. Whenever a remote user enters text messages, the accumulated texts are converted to an HTML file for a local browser to display. 
	For example, suppose that the user enters ``Hi, Bobby!'', ``Your nose is shiny.'', and ``Where is John?'', in sequence. Then the following HTML files will be generated in sequence:
	
	\begin{itemize}
	\item \ul{} 
	\item \ul{\li{Hi, Bobby!}} 
	\item \ul{\li{Hi, Bobby!} \li{Your nose is shiny.}} 
	\item \ul{\li{Hi, Bobby!} \li{Your nose is shiny.} \li{Where is John?}}
	\end{itemize}
\ \\
Then the browser will display them  one by one replacing each HTML file with the next one. 

	A program that behaves in this way may have type $Stream \ String \funL{\client} Stream \ HTML$ on a hypothetical runtime system, e.g., 
	\[
	\lamL{\client}{prg}{\ browser \ (prg \ keyboard)}
	\]
	where $keyboard$ offers the program, $prg$, a stream of strings that the user enters, and $browser$ displays a stream of HTMLs generated by the program. 

	Thus a simple minded multi-tier program is a conversion function of a string stream into an HTML stream.
	Such a program may be written by composing a server function and a client function by a {\it cross-tier composition function} instantiated with the client location and the server location,  $\circ{[\server,\client]}$ (or written as $\circ_{[\server,\client]}$ for readability), as this. \ \\
	
	\begin{tabular}{l c l}
	$ program =  (foldl_{[\client]} \ f \  [\ ]) \ \ (\circ_{[\server,\client]}) \ \ (foldl_{[\server]} \  g \  [ \  [\ \! ] \ ])$  \\[0.1cm]
	\end{tabular}
	
	\begin{tabular}{l}
	\ \ where \\
	\ \ \ \ $f = \lamL{\client}{htmls}{
					\lamL{\client}{strs}{\ htmls \ {\scriptstyle ++}_{[\client]} } }$ \\
	\ \ \ \ \ \ \ \ \ \ \ \ \ 
	\ \ \ \ \ \  
						$[\ \ul{ \ \ \ map_{[\client]} \  (\lamL{\client}{str}{\ \li{str}}) \ strs \ \ \ } \ ]$ \\[0.1cm] % 
	\ \ \ \ $g = \lamL{\server}{strs}{
					\lamL{\server}{str}{\ strs \ {\scriptstyle ++}_{[\server]} \ [ \ \ \  last_{[\server]} \ strs \ {\scriptstyle ++}_{[\server]} \ [str]  \ \ \ ]} }$ \\[0.1cm] % 						
	\end{tabular}

	Note that the example uses $[$ and $]$ to denote a stream. It also assumes the standard stream functions extended with a single polymorphic location abstraction: the left fold function, $foldl$, the stream concatenation function, (${\scriptstyle ++}$), the map function, $map$, and the last element selection function, $last$. 
%	For example, the map function over streams has type
%	\[
%	\forall l. \forall \alpha\beta. (\alpha\funL{l}\beta) \funL{l} (Stream \ \alpha\ \funL{l} \ Stream \ \beta),
%	\]
%	and its term will be explained in Section \ref{subsubsec:monomorphizationexamples}.
	
	Also note that as the multi-tier programming languages such as Links, Ur/Web, and so on provide one's own convenient notation for HTML, the example makes use of such notation as $\ul{\cdots}$ and $\li{\cdots}$ to construct HTML elements. They can be viewed as values of a special base type $HTML$.
	
	In the structure of the program, the server function, $foldl_{[\server]} \  g \  [ \  [\ \! ] \ ]$, is of type $Stream \ String \funL{\server} Stream \ Text$ where $Text$ is the type for a finite list of strings, and it generates a sequence of texts accumulated up to whenever a user enters a string. The client function, $foldl_{[\client]} \ f \  [\ ]$, of type $Stream \ Text \funL{\client} Stream \ HTML$ takes the sequence of texts through the composition function. Each text, $[str_1, \cdots, str_n]$, is converted into an HTML, 	$\ul{\li{str_1}, \cdots, \li{str_n} }$ for the browser to display, and the conversion is repeated over the sequence to update one HTML with the next one. 

	Now we are ready to review how the notion of polymorphic location is useful in writing multi-tier programs. First of all, polymorphic location functions are extensively used in the example. It is straightforward to use them at different locations just by instantiating them with the locations. The example requires no duplicate codes that would otherwise be written in the monomorphically typed RPC calculus, for example, for $foldl$, (${\scriptstyle ++}$), $map$, and $last$. 
	The resulting example looks more like a single computer program when all location applications are ignored, which is an ultimate goal in the multi-tier programming languages.  
	
	Second of all, it is also easy to compose two differently located functions as shown by $\circ[\server,\client]$ in the example. 
	The term and the type of the cross-tier composition function, ($\circ$), can be defined as follows.

	\begin{tabular}{l c l}
	$\circ$ & $:$ & $\forall l_1, l_2. \forall\alpha,\beta,\gamma. (\beta\funL{l_2}\gamma) \funL{l_2} (\alpha\funL{l_1}\beta) \funL{l_2} (\alpha\funL{l_2}\gamma)$ \\
	$\circ$ & $=$ & $\Lambda l_1, l_2. \Lambda\alpha,\beta,\gamma. \lamL{l_2}{f}{\lamL{l_2}{g}{\lamL{l_2}{x}{\ f \ (g \ x)}}}$ \\
	\end{tabular} 

\ \\
Then $f \ \circ_{[\server,\client]} \ g$ composes a server function $g$ and a client function $f$ with data flow from the server to the client. $f \ \circ_{[\client,\server]} \ g$ does another cross-tier composition with the reverse data flow. Of course, it is also easy to express the local composition by $\circ[\client,\client]$ and $\circ[\server,\server]$. 
	Without the notion of polymorphic location, four different located functions would be written, which is more laborious task. 
	
\subsection{A monomorphization translation of the polymorphic RPC calculus}

	Once programmers have the polymorphic RPC calculus in hand to be able to write RPC terms succinctly with the notion of polymorphic location, the next step is to design a method to compile to the client-server model.
	Instead of reinventing the wheel, we choose to make use of the two slicing compilation methods for the typed RPC calculus after translating polymorphic RPC terms into typed RPC terms. To realize this strategy, we need to translate all features of polymorphic locations away completely. This monomorphization translation will be discussed in this section. 

	A key idea behind the monomorphization translation is to interpret a location abstraction as a pair of its client version and its server version, and to regard a location application as a projection of the pair according to the location to be applied.  For example, an identity function $id$, $\Lambda l. \lamL{l}{x}{x}$, is translated as $(\lamL{\client}{x}{x}, \lamL{\server}{x}{x})$. Also, $id[\client]$ and $id[\server]$ are translated as $\pi_1(\mono{id})$ and $\pi_2(\mono{id})$, respectively, where $\mono{M}$ is a translation of $M$. A systematic translation like this allows the monomorphic typed RPC to pretend to have the notion of polymorphic locations without explicit location abstraction and location application. 
	
	Note that this pair-based translation scheme can potentially produce an exponentially large term. We believe that this would not be a big problem in practice as long as the nesting depth of location lambdas is shallow. This issue will be discussed more later. 
	
	For the target  of the translation, we extend the typed RPC calculus with pairs $(M,N)$ and projections $\pi_i(M)$ in the standard way as shown in Figure \ref{fig:rpcwithpair} and Figure \ref{fig:rpctysystemwithpair}. The extended type language includes pair types $A \times B$. From now on, we call this extension just the typed RPC calculus {\rpc}. 
	
	%Note that a set of free location variables over pair types can be defined as: $flv(A \times B)=flv(A) \cup flv(B)$. 
	
\begin{figure}[t]
\begin{tabular}{ l  l  l  c  l    c  l  c  l  c  l c l }
\multicolumn{13}{l}{\textbf{Syntax}} \\
  \ \ \ 
            & Location 		& $a,b$ 		& $::=$ 	& $\client$	& $|$		& $ \server$		 \\
            & Term 	& $L,M,N$	& $::=$ 	& $V$			& $|$		& $L \ M$	 & 	$|$		&	$M[A]$					\\
            &			&				& $|$	& $(M,N)$		& $|$ 	& $\pi_i(M)$ \\
            & Value    			& $V,W$		& $::=$ 	& $x$			& $|$		& $ \lambda^{a} x.M$ 	&		$|$		& $\Lambda\alpha.V$	& $|$		& $ (V,W)$
\\[0.1cm]
\multicolumn{13}{l}{\textbf{Semantics}} \\
 % & 
 \multicolumn{13}{c}{
 \ \  
  	\mbox{
       \begin{prooftree}
       		\Infer0[(Abs)]{ \evalRPC{\lamL{b}{x}{M}}{a}{\lamL{b}{x}{M} }}
	\end{prooftree}
	\ \ \ 
	\begin{prooftree}
  		\Hypo{ \evalRPC{L}{a}{\lamL{b}{x}{N}} }
  		\Hypo{ \evalRPC{M}{a}{W} }
  		\Hypo{ \evalRPC{N\subst{W}{x}}{b}{V}  }
  		\Infer3[(App)]{ \evalRPC{L \ M}{a}{V}  }
	\end{prooftree} 
  	}
    } 
\\[0.5cm]
 %&
  \multicolumn{13}{c}{
  \ \ \ 
  	\mbox{
       \begin{prooftree}
       		\Infer0[(Tabs)]{ \evalRPC{\Lambda\alpha.V}{a}{\Lambda\alpha.V }}
	\end{prooftree}
	\ \ \ \ \ 
	\begin{prooftree}
  		\Hypo{ \evalRPC{M}{a}{\Lambda\alpha.V} }
  		\Infer1[(Tapp)]{ \evalRPC{M[B]}{a}{V\subst{B}{\alpha}}  }
	\end{prooftree} 
  	}
    } 
\\[0.5cm]
% & 
 \multicolumn{13}{c}{
 \ \  
  	\mbox{
       \begin{prooftree}
       	\Hypo{ \evalRPC{L}{a}{V} }
       	\Hypo{ \evalRPC{M}{a}{W} }
       	\Infer2[(Pair)]{ \evalRPC{(L,M)}{a}{(V,W) }}
	\end{prooftree}
	\ \ \ \ \ 
	\begin{prooftree}
  		\Hypo{ \evalRPC{M}{a}{(V_1,V_2)}  \ \ \ i\in\{1,2\}}
  		\Infer1[(Proj-i)]{ \evalRPC{\pi_i(M)}{a}{V_i}  }
	\end{prooftree} 
  	}
    } 
\end{tabular}
\caption{The typed RPC calculus extended with pairs}
\label{fig:rpcwithpair}
\end{figure}

\begin{figure}[h]
\begin{tabular}{l l l l l l l l l l l l l }
\multicolumn{13}{l}{\textbf{Types}} \\
   \ \ \  & Type & $A,B,C$ & $::=$ & $base$ & $ | $ & $ A\funL{a}B$  & $|$ & $\alpha$ & $|$ & $\forall\alpha.A$ & $ | $ & $ A \times B$ \\[0.1cm]
\multicolumn{13}{l}{\textbf{Typing Rules}} \\
&
 \multicolumn{12}{c}{  
  \mbox{
  	\begin{prooftree}
		\Hypo{  \tyenv(x)=A }
		\Infer[left label=(T-Var)]1{ \typing{\tyenv}{a}{x}{A} }
	\end{prooftree}
		\ \ \ 
	\begin{prooftree}
		\Hypo{ \typing{\tyenvExt{x}{A}}{b}{M}{B} }
		\Infer[left label=(T-Abs)]1{ \typing{\tyenv}{a}{\lamL{b}{x}{M}}{A\funL{b}B} } 
	\end{prooftree}
    } 
    }
\\[0.5cm] 
 &
 \multicolumn{12}{c}{  
  \mbox{
	\begin{prooftree}
		\Hypo{  \typing{\tyenv}{a}{L}{A\funL{b}B } }
		\Hypo{  \typing{\tyenv}{a}{M}{A} }
		\Infer[left label=(T-App)]2{ \typing{\tyenv}{a}{L \ M}{B}   }
	\end{prooftree}
    } 
    }
\\[0.5cm]
& 
 \multicolumn{12}{c}{  
 \mbox{
	\begin{prooftree}
		\Hypo{  \typing{\tyenv,\alpha}{a}{V}{A} }
		\Infer[left label=(T-Tabs)]1{ \typing{\tyenv}{a}{\Lambda\alpha.V}{\forall\alpha.A}   }
	\end{prooftree}
	\ \ \ \ \
	\begin{prooftree}
		\Hypo{  \typing{\tyenv}{a}{M}{\forall\alpha.A} }
		\Infer[left label=(T-Tapp)]1{ \typing{\tyenv}{a}{M[B]}{A\subst{B}{\alpha}}   }
	\end{prooftree}
	}
    }
\\[0.5cm] 
 & 
\multicolumn{12}{c}{  
  \mbox{
	\begin{prooftree}
		\Hypo{ \typing{\tyenv}{a}{L}{A} }
		\Hypo{ \typing{\tyenv}{a}{M}{B} }
		\Infer[left label=(T-Pair)]2{ \typing{\tyenv}{a}{(L,M)}{ A \times B }} 
	\end{prooftree}
    } 
    }
\\[0.5cm] 
&
\multicolumn{12}{c}{  
  \mbox{
	\begin{prooftree}
		\Hypo{ \typing{\tyenv}{a}{M}{A_1 \times A_2} \ \ \ i\in\{1,2\} }
		\Infer[left label=(T-Proj-i)]1{ \typing{\tyenv}{a}{\pi_i(M)}{ A_i } } 
	\end{prooftree}
    } 
    }
\end{tabular}
\caption{A type system for the extended typed RPC calculus}
\label{fig:rpctysystemwithpair}
\end{figure}

	We present a monomorphization translation of {\polyrpc} into {\rpc} as shown in Figure \ref{fig:translation}. The translation has two parts. The first part is for types. 
	The type translation of $base$ is $base$. 
	We translate a monomorphic function type with a location $a$ element-wise leaving the location as it is.  Any attempt to translate a function type with a location variable is undefined. We prevent this undesirable translation by systematically substituting $\client$ or $\server$  for every location variable that we meet during a translation and by starting from a closed type with no free location variables in it. The translation is designed to operate on closed types only with location constants all the time.
	The translation of polymorphic location types  is $\mono{\forall l.A} = (\mono{A\subst{\client}{l}}, \mono{A\subst{\server}{l}})$, which realizes the pair-based interpretation of location polymorphism as well as the invariant of location constants. 
	
	The second part is a translation of terms. 
	The translations of variable, $\lambda$-abstraction, $\lambda$-application, type abstraction, and type application are done element-wise. 
	The translation of location abstraction is defined as a pair-based interpretation with the first element for client and the second element for server: $\mono{\Lambda l.V} = (\mono{V\subst{\client}{l}}, \mono{V\subst{\server}{l}})$. The translation maintains an invariant that whenever a term $\Lambda l. V$ is closed, the translated term also remains closed. This is so by systematically substituting $\client$ or $\server$  for every occurrence of the location variable $l$ in the body $V$.
	The translation of location application with $\client$ or $\server$ is a projection of a pair that represents a polymorphic term. It is undefined when a location variable $l$ appears as an argument in the location application. Our design associates the client version of the polymorphic term with the first element of the pair, and it does the server version with the second element. Therefore, $\mono{ M[a] } = \pi_i( \mono{M} )$ where $i=1$ if $a=\client$ and   $i=2$ if $a=\server$.
	
	The monomorphization translation can be naturally extended for typing environments as: $\mono{\{\alpha_1,\cdots,\alpha_k,x_1:A_1, \cdots, x_n:A_n\}} = \{ \alpha_1,\cdots,\alpha_k,x_1:\mono{A_1}, \cdots, x_n:\mono{A_n}\}$. This translation is undefined if the typing environment contains any location variables or types associated with variables have any free location variables. 
	
	\begin{figure}[t]
\begin{tabular}{l l l l l l l l l }
\multicolumn{9}{l}{\textbf{Translation: types}} \\
\ \ 
	$\mono{base}=base$  & 
	$\mono{A \funL{a}B}=\mono{A}\funL{a}\mono{B}$ & 
	$\mono{\forall l.A} = \mono{ A\subst{\client}{l} } \times \mono{A\subst{\server}{l} }$
\\
\ \ 
	$\mono{\alpha}=\alpha$ &
	$\mono{\forall\alpha.A}=\forall\alpha.\mono{A}$
\\[0.1cm]
\multicolumn{9}{l}{\textbf{Translation: terms}} \\
\ \ 
	$\mono{x}=x$ &
	$\mono{\lamL{a}{x}{M}} = \lamL{a}{x}{\mono{M}}$ &
	$\mono{\Lambda l. V} = (\mono{V\subst{\client}{l}}, \mono{V\subst{\server}{l}})$ & 
\\[0.1cm]
\ \ 
	$\mono{L \ M} = \mono{L} \ \mono{M}$ &
	$\mono{M[\client]}  = \pi_1(\mono{M})$ &
	$\mono{M[\server]} = \pi_2(\mono{M})$
\\
\ \ 
	$\mono{\Lambda\alpha.V}=\Lambda\alpha.\mono{V}$ &
	$\mono{M[B]}=\mono{M}[\mono{B}]$
\end{tabular}
\caption{A translation of the polymorphic RPC calculus into the RPC calculus}
\label{fig:translation}
\end{figure}

\subsubsection{Examples of the monomorphization translation}
\label{subsubsec:monomorphizationexamples}

	Let us discuss a few examples of the translation.
	The simplest one is with an identity function as: 
		\begin{eqnarray*}
		\mono{\Lambda l. \lamL{l}{x}{x}} 
		&=& (\mono{(\lamL{l}{x}{x})\subst{\client}{l}}, \mono{(\lamL{l}{x}{x})\subst{\server}{l}})  \\
		&=& (\lamL{\client}{x}{\mono{x}}, \lamL{\server}{x}{\mono{x}}) \\
		&=& (\lamL{\client}{x}{x}, \lamL{\server}{x}{x}). 
		\end{eqnarray*}
	Next are more complex examples about map functions over lists. Suppose $M_{X}$ is a body of a map function parameterized by location arguments $X$ as: $case \ xs \ of \ \{ \ nil \Rightarrow nil; \ cons \ y \ ys \Rightarrow cons \ (f \ y) \ (map [X] \ f \ ys) \ \}$ under some syntactic extension with case expression and  with list constructors such as $nil$ and $cons$. Also assume that $[A]$ is the list type over element type $A$.
	
	Then a map function of type $\forall l. (A \funL{l} B) \funL{l} ([A] \funL{l} [B])$ can be defined by a recursive let construct as:
	 \[
	\texttt{ letrec } \ map = \Lambda{l}.\lamL{l}{f}{ \lamL{l}{xs}{ M_{l} } } \ \texttt{ in } \ \cdots
	\]
	The translation of this map function proceeds as 
	\begin{eqnarray*}
	\mono{map} &=& ( \mono{ (\lamL{l}{f}{ \lamL{l}{xs}{ M_{l}} }) \subst{\client}{l} } ,  \mono{ (\lamL{l}{f}{ \lamL{l}{xs}{ M_{l}} }) \subst{\server}{l} } ) \\
	&=& ( \lamL{\client}{f}{ \mono{\lamL{\client}{xs}{ M_{\client}} } } , \lamL{\server}{f}{ \mono{\lamL{\server}{xs}{ M_{\server}} } } ) \\
	&=& ( \lamL{\client}{f}{ \lamL{\client}{xs}{ \mono{M_{\client}} } } , \lamL{\server}{f}{ \lamL{\server}{xs}{ \mono{M_{\server}} } } ).
	\end{eqnarray*} 
	To finish the translation of a recursive function such as $map$, we will ultimately need something like tying a knot  over $map[X]$ in the recursive call $map[X] \ f \ ys$. 
	In the first element of the translated pair, $map[\client]$ in $M_{\client}$ is translated as $\pi_1(\mono{map})$, and,
	in the second element, $map[\server]$ in $M_{\server}$ is translated as $\pi_2(\mono{map})$. 
	Here, instead of repeating $\mono{map}$ from the beginning, the translation stops here, and it just refers to the result of the beginning translation of $map$ lazily to construct a recursive function. Then the translation continues over the rest of the body. As a result, we will get a pair of the client map function and the server map function. 
	
	Let us consider another map function of type $\forall l_1.\forall l_2.\forall l_3. (A \funL{l_3} B) \funL{l_1} ([A] \funL{l_2} [B])$ can be defined as:
	 \[
	 \texttt{ letrec } \ map = \Lambda{l_1}.\Lambda{l_2}.\Lambda{l_3}.map_0  \ \texttt{ in } \ \cdots
	\]
	where $map_0$ is $ \lamL{l_1}{f}{ \lamL{l_2}{xs}{ M_{l_1 l_2 l_3} }}$ and multiple location applications $M[\Loc_1 \cdots \Loc_k]$ is defined as $M[\Loc_1]\cdots[\Loc_k]$. The translation $\mono{map}$ becomes
	\[ 	
	\left(
	\mbox{
		\begin{tabular}{l }
		$\left( 
			\mbox{
			\begin{tabular}{l }
			$( \mono{map_0\{\client/l_1,\client/l_2,\client/l_3\}},  \mono{map_0\{\client/l_1,\client/l_2,\server/l_3\}} ),$ \\
			$( \mono{map_0\{\client/l_1,\server/l_2,\client/l_3\}},  \mono{map_0\{\client/l_1,\server/l_2,\server/l_3\}} )$
			\end{tabular}
			}
		\right),$ \\ 
		$\left( 
			\mbox{
			\begin{tabular}{l }
			$( \mono{map_0\{\server/l_1,\client/l_2,\client/l_3\}},  \mono{map_0\{\server/l_1,\client/l_2,\server/l_3\}} ),$ \\
			$( \mono{map_0\{\server/l_1,\server/l_2,\client/l_3\}},  \mono{map_0\{\server/l_1,\server/l_2,\server/l_3\}} )$
			\end{tabular}
			}
		\right)$
		\end{tabular}
	}
	\right)
	\]
	which looks like a full binary tree with eight leaves. As you see, the monomorphization translation could generate exponentially many instances. 

	Let us discuss this problem in two ways. 
	First, we believe that the first version of the map function would be typical in multi-tier programs by annotating the same location variable on all function types of a polymorphic type.  For example, the Links compiler already performs something similar to the monomorphization for the standard library such as list utilities, which is common between the client and the server programs. It does code duplication for each common function as done with the first version, not for each argument of the functions as done with the second version. Our translation method can be regarded as a type-theoretic basis for the ad hoc code duplication in the Links compiler. In addition to this, the polymorphic RPC calculus offers a type-based method to control code duplication in a fine-grained way. 
	
	Note that the Eliom \cite{Radanne2017,Radanne:2018:TWP:3184558.3185953} compiler also uses code duplication to support a kind of location polymorphism, what they call {\it shared sections} in the expression-level and {\it mixed modules} in the module-level. In the related work, we will discuss these features in detail. 

%	Second, the compiler could be careful to generate instances on demand hoping that the number of necessary instances may be small, though this may need a  whole-program  dead  code analysis. Since this research on polymorphic location is yet in an early stage, we need further experiences on this matter.
  
	Second, we could take an alternative approach  of dynamically passing locations even in separate client and server programs, which is radically different from the proposed static compilation of location polymorphism. This comes from the idea of compiling polymorphism using intensional type analysis on types arranged dynamically \cite{Harper:1995:CPU:199448.199475}. It would solve the potential code explosion problem completely at the expense of runtime cost for maintaining location information and deciding whether to do local or remote calls dynamically. We will discuss this approach later.
	
%	The last example is a map function with polymorphic recursion.  
%	Let us consider another map function of type $\forall l_1.\forall l_2.\forall l_3. (A \funL{l_3} B) \funL{l_1} ([A] \funL{l_2} [B])$ can be defined as:
%	 \[
%	 letrec \ map = \Lambda{l_1}.\Lambda{l_2}.\Lambda{l_3}.map_1 \ in \ \cdots
%	\]
%	where $map_1$ is $ \lamL{l_1}{f}{ \lamL{l_2}{xs}{ M_{l_2 l_3 l_1} }}$ shifting location arguments to the left by one for each recursion. Interestingly, the patterns of knot tyings to construct a recursive function are complex in the polymorphic recursion version. For example, $\mono{map_0\{\client/l_1,\client/l_2,\server/l_3\}}$, $\mono{map_0\{\client/l_1,\server/l_2,\client/l_3\}}$, and $\mono{map_0\{\server/l_1,\client/l_2,\client/l_3\}}$ will form a chain of a recursive function. However, we are unclear yet what this means in practice, which is left for a future study. 

\subsubsection{Type correctness}

	Now we are ready to formulate the type correctness of the monomorphization translation as Theorem \ref{theorem:typecorrectness} saying that every closed well-typed term in {\polyrpc} is translated to a well-typed term in {\rpc}. We prove this theorem by Lemma \ref{lemma:typecorrectness}, which generalizes the formulation of the theorem to allow to have non-empty typing environments but only with free variables and free type variables, saying if \ $\typing{\tyenv}{a}{M}{A}$  then $\typing{\mono{\tyenv}}{a}{\mono{M}}{\mono{A}}$. 	
	This lemma should maintain an invariant that the translation always sees type environments, terms, and types only with  location constants. This invariant is sufficient to make sure that the three translations in the lemma are well-defined. 
	
	For this invariant, we need to define a set of free location variables for terms, $flv(M)$, as well. 
	Here is a definition:	
	\begin{eqnarray*}
	flv(x) &=& \emptyset \\
	flv(\lamL{\Loc}{x}{M}) &=& flv(\Loc)\cup flv(M) \\
	flv(\Lambda \alpha. V) &=& flv(V) \\
	flv(\Lambda l. V) &=& flv(V)\backslash \{l\} \\
	flv(L \ M) &=& flv(L) \cup flv(M) \\
	flv(M [A]) &=& flv(M) \cup flv(A) \\
	flv(M [\Loc]) &=& flv(M) \cup flv(\Loc)
	\end{eqnarray*}
	
	The proof of Lemma \ref{lemma:typecorrectness} is done by induction on the height of a typing derivation for a term that the translation is to apply to. The invariant of location constants plays a role to ensure that the translation is well-defined for each subterm. The inductive proof works for all the cases except the case of (T-Tapp) and the case of (T-Labs). To finish the case of (T-Tapp), we use Lemma \ref{lemma:typesubstitutionmono}. This lemma  turns a  type substituted after a translation, which is obtained by the inductive hypothesis, into a translation of a substituted type, which is necessary to prove the case. 
	
	The case of (T-Labs) is more interesting. Dealing with $\Lambda l.V$, it needs two hypotheses over $V\subst{\client}{l}$ and $V\subst{\server}{l}$, not a hypothesis over $V$ that is required by the plain inductive proof but violates the invariant because of the potential occurrence of a free location variable $l$. 
	Lemma \ref{lemma:locationpolymorphism} allows the proof to use the two hypotheses as follows. 
	
	According to Lemma \ref{lemma:locationpolymorphism}, given a typing derivation possibly with the use of free location variables, one can build a new typing derivation with the similar structure and the same height by substituting any of $\client$ or $\server$ for the occurrences of the free location variables in the typing derivation. 
	This lemma accounts for location polymorphism in the polymorphic RPC calculus, which is why we call it the location polymorphism lemma. 
	
	The proof of Lemma \ref{lemma:typesubstitutionmono} and Lemma \ref{lemma:locationpolymorphism}  is available in the appendix.

%%%
\begin{lemma}[Type substitution over type under monomorphization] $\mono{A}\subst{\mono{B}}{\alpha}$  $=$ $\mono{A\subst{B}{\alpha}}$.
\label{lemma:typesubstitutionmono}
\end{lemma}

\begin{lemma}[Location polymorphism] Suppose $\tyenv=\{l_1,\cdots, l_n\}\cup\tyenv_0$  such that $\tyenv_0$ has no location variables. If \ $\typing{\tyenv}{\Loc}{M}{A}$ then $(\typing{\tyenv_0}{\Loc}{M}{A})\{a_1/l_1,\cdots, a_n/l_n\}$ for any $a_1,\cdots, a_n$ with the same height.
\label{lemma:locationpolymorphism}
\end{lemma}

\begin{lemma} Given $flv(\tyenv) \cup flv(M) \cup flv(A) = \emptyset$, if \ $\typing{\tyenv}{a}{M}{A}$  then $\typing{\mono{\tyenv}}{a}{\mono{M}}{\mono{A}}$.
\label{lemma:typecorrectness}
\end{lemma}
\begin{proof}
	We prove this lemma by induction on the height of the typing derivation. A base case uses (T-Var), and  inductive cases involve the other kinds of typing rules. 
	
{\bf (T-Var):} 
	$\mono{\tyenv}$ and $\mono{A}$ are well-defined since $FLV(\tyenv)=\emptyset$ and $FLV(A)=\emptyset$. By the typing derivation, $\tyenv(x)=A$. This implies (1):$\mono{\tyenv}(x)=\mono{A}$. By (T-Var) with (1), $\typing{\mono{\tyenv}}{a}{x}{\mono{A}}$, which proves this case since $\mono{x}=x$. 
	
{\bf (T-Abs):}
	$\typing{\tyenv}{a}{\lamL{\Loc}{x}{M_0}}{A_0\funL{\Loc} B_0}$, (1):$\typing{\tyenvExt{x}{A_0}}{\Loc}{M_0}{B_0}$ where $M=\lamL{\Loc}{x}{M_0}$ and $A=A_0\funL{\Loc} B_0$. It is easy to verify (2):$flv(\tyenvExt{x}{A_0})\cup flv(M_0)\cup flv(B_0)=\emptyset$ by the hypothesis on free location variables.  
	
	I.H. can be applied to to (1) and (2) since $\Loc=b$ for some $b$; $flv(A_0\funL{\Loc} B_0)=\emptyset$ implies $flv(\Loc)=\emptyset$. This implies (3):$\typing{ \mono{\tyenvExt{x}{A_0}} }{\Loc}{ \mono{M_0} }{ \mono{B_0} }$. 
	
	By (T-Abs) with (3), (4):$\typing{ \mono{\tyenv} }{a}{\lamL{\Loc}{x}{ \mono{M_0} }}{ \mono{A_0}\funL{\Loc}\mono{B_0} }$. By the definition of the translation, (4) implies $\typing{ \mono{\tyenv} }{a}{\mono{ \lamL{\Loc}{x}{ M_0 } } }{ \mono{A_0\funL{\Loc}B_0} }$ since $\Loc=b$.

{\bf (T-App):}
	$\typing{\tyenv}{a}{L \ M_0}{A}$, (1):$\typing{\tyenv}{a}{L}{B\funL{\Loc'}A}$, and (2):$\typing{\tyenv}{a}{M_0}{B}$ where $M=L \ M_0$. 
	
	By well-formed typing judgments, $flv(B)=\emptyset$ and $flv(\Loc')=\emptyset$. 
%	though the type system for {\polyrpc} is defined a bit loosely to have a potential to miss these. But it can be made more precise to capture them by introducing a set of (usable) free location variables $\Delta$ to typing judgments as $\typing{\Delta;\tyenv}{a}{M}{A}$. In (T-Labs), $\Delta$ forms only from location abstractions by adding their location variable to $\Delta$. We then introduce to (T-App) a constraint as $flv(B) \subseteq \Delta$ and $flv(\Loc') \subseteq \Delta$. By revising the theorem to have an invariant that $\Delta$ remains empty, it can be enforced that $B$ and $\Loc'$ are chosen in (T-App) with no unbound location variables occurring in them. 
	
	By the argument stated above, we can safely assume (3):$flv(\tyenv)\cup flv(L)\cup flv(B\funL{\Loc'}A)=\emptyset$. Also, we have (4):$flv(\tyenv)\cup flv(M_0)\cup flv(B)=\emptyset$.
	
	By applying I.H. to (1) and (3), (5):$\typing{ \mono{\tyenv} }{a}{ \mono{L} }{ \mono{B\funL{\Loc'}A} }$.
	
	By applying I.H. to (2) and (4), (6):$\typing{ \mono{\tyenv} }{a}{ \mono{M_0} }{ \mono{B} }$.
	
	Since $flv(\Loc')=\emptyset$,  $\Loc'=b$ for some $b$. Therefore, $\mono{B\funL{\Loc'}A}= \mono{B}\funL{\Loc'} \mono{A}$.  By (T-App) with (5) and (6), $\typing{ \mono{\tyenv} }{a}{ \mono{L} \mono{M_0} }{ \mono{A} }$, which proves this case by $\mono{M}=\mono{L M_0}=\mono{L} \mono{M_0}$.

{\bf (T-Tabs):}
	$\typing{\tyenv}{a}{\Lambda\alpha.V_0}{\forall\alpha.A_0}$, and (1):$\typing{\tyenvExtWith{\alpha}}{a}{V_0}{A_0}$ where $M=\Lambda\alpha.V_0$ and $A=\forall\alpha.A_0$.

	(2):$flv(\tyenvExtWith{\alpha}) \cup flv(V_0) \cup flv(A_0)=\emptyset$. 
	
	By applying I.H. to (2), $\typing{ \mono{\tyenvExtWith{\alpha}} }{a}{\mono{V_0}}{\mono{A_0}}$, which is (3):$\typing{\mono{\tyenv},\alpha}{a}{\mono{V_0}}{\mono{A_0}}$ by the def. of the translation.
	
	By (T-Tabs) with (3), $\typing{\mono{\tyenv}}{a}{\Lambda\alpha.\mono{V_0}}{\forall\alpha.\mono{A_0}}$. This is the same as $\typing{\mono{\tyenv}}{a}{\mono{\Lambda\alpha.V_0}}{\mono{\forall\alpha.A_0}}$ by the def. of the translation, which proves this case. 
	
{\bf (T-Tapp):} 
	$\typing{\tyenv}{a}{M_1[B]}{A_0\subst{B}{\alpha}}$, and (1):$\typing{\tyenv}{a}{M_1}{\forall\alpha.A_0}$ where $M=M_1[B]$ and $A=A_0\subst{B}{\alpha}$. 
	
	(2):$flv(\tyenv) \cup flv(M_1) \cup flv(\forall\alpha.A_0) = \emptyset$. 
	
	By I.H. with (2), $\typing{ \mono{\tyenv} }{a}{ \mono{M_1} }{ \mono{\forall\alpha.A_0} }$, which is (3):$\typing{ \mono{\tyenv} }{a}{ \mono{M_1} }{ \forall\alpha.\mono{A_0} }$ by the def. of the translation. 
	
	By (T-Tapp) with(3),(4):$\typing{ \mono{\tyenv} }{a}{ \mono{M_1}[ \mono{B} ] }{ \mono{A_0}\subst{\mono{B}}{\alpha} }$. 
	
	By the lemma \ref{lemma:typesubstitutionmono} with (4), (5):$\typing{ \mono{\tyenv} }{a}{ \mono{M_1}[ \mono{B} ] }{ \mono{A_0\subst{B}{\alpha}} }$, which proves this case by the def. of the translation.
	
{\bf (T-Labs):}
	$\typing{\tyenv}{a}{\Lambda l.V_0}{\forall l.A_0}$, and (1):$\typing{\tyenvExtWith{l}}{a}{V_0}{A_0}$ where $M=\Lambda l.V_0$ and $A=\forall l.A_0$. 
	
	By applying the lemma \ref{lemma:locationpolymorphism} to (1), (2): $\typing{\tyenv}{a}{V_0\subst{a}{l}}{A_0\subst{a}{l}}$ for all $a$. Also, (3):$flv(\tyenv)\cup flv(V_0\subst{a}{l})\cup flv(A_0\subst{a}{l})=\emptyset$ since $l$ is the only free location variable. 
	
	By applying I.H. to (2) and (3), (4):$\typing{ \mono{\tyenv} }{a}{ \mono{V_0\subst{a}{l}} }{ \mono{A_0\subst{a}{l}} }$ for $a=\client,\server$.
	
	By (T-Pair) in the type system for {\rpc} with two premises, one with (4) and $a=\client$ and the other with (4) and $a=\server$, $\typing{ \mono{\tyenv} }{a}{ ( \mono{V_0\subst{\client}{l}}, \mono{V_0\subst{\server}{l}} ) }{  \mono{A_0\subst{\client}{l}} \times \mono{A_0\subst{\server}{l}}  }$. By the definition of the translation, $\typing{ \mono{\tyenv} }{a}{ \mono{\Lambda l.V_0} }{ \mono{\forall l.A_0} }$, which proves this case. 

{\bf (T-Lapp):}
	$\typing{\tyenv}{a}{L[\Loc']}{A_0\subst{\Loc'}{l}}$, and (1):$\typing{\tyenv}{a}{L}{\forall l.A_0}$ where $M=L[\Loc']$, $A=A_0\subst{\Loc'}{l}$. $\Loc'=b$ for some $b$ since $flv(\Loc')=\emptyset$ by the hypothesis on free location variables. Also, (2):$flv(\tyenv)\cup flv(L)\cup flv(\forall l.A_0)=\emptyset$ since $flv(A_0\subst{\Loc'}{l})=\emptyset$ and so $l$ can be the only free location variable in $A_0$. 
	
	By applying I.H. to (1) and (2), $\typing{ \mono{\tyenv} }{a}{ \mono{L} }{ \mono{\forall l.A_0} }$, which is (3): $\typing{ \mono{\tyenv} }{a}{ \mono{L} }{  \mono{A_0\subst{\client}{l}} \times \mono{A_0\subst{\server}{l}}  }$ by the definition of the translation.
	
	By (T-Proj-i) with (3), $\typing{ \mono{\tyenv} }{a}{ \pi_i(\mono{L}) }{ \mono{A_0\subst{\Loc'}{l}} }$ where  $i=1$ if $\Loc'=\client$, and $i=2$ if $\Loc'=\server$, which is (4):$\typing{ \mono{\tyenv} }{a}{ \mono{L[\Loc']} }{ \mono{A_0\subst{\Loc'}{l}} }$ by the definition of the translation. (4) proves this case.
\end{proof}

\begin{theorem}[Type correctness of the monomorphization translation] For a closed term $M$, if \ $\typing{\emptyset}{a}{M}{A}$ then $\typing{\emptyset}{a}{\mono{M}}{\mono{A}}$.
\label{theorem:typecorrectness}
\end{theorem}
\begin{proof} This theorem is proved by the lemma \ref{lemma:typecorrectness} as $M$ is a top-level closed term and so $flv(\emptyset) \cup flv(M) \cup flv(A) = \emptyset$.
\end{proof}

%%%

\subsubsection{Semantic correctness}

	This time we formulate a stronger property, the semantic correctness of the monomorphization translation, as Theorem \ref{thm:semanticcorrectness}: for a well-typed closed term, the evaluation of the term in {\polyrpc} is preserved in {\rpc} under the translation. We prove this theorem by induction on the height of the evaluation derivation for the closed well-typed term. Two of the three base cases are (Abs) and (Tabs), which are proved immediately. The case (App) of the proof uses Lemma \ref{lemma:substitutionmono} saying the translation of a substitution is the same as a substitution of the translated terms.  Similarly, the case (Tapp) uses Lemma \ref{lemma:typesubstitutiontermmono} to turn a translated term substituted with a translated type into a translation of a term substituted with a type. The most interesting cases are (Labs) and (Lapp).
	
	(Labs) is the remaining one of the three base cases, and it needs to prove that if $\evalRPC{\Lambda l.V}{a}{\Lambda l.V}$, then $\evalRPC{\mono{\Lambda l.V}}{a}{\mono{\Lambda l.V}}$. The proof in this case needs the second induction because $\mono{\Lambda l.V}$ becomes translated into a possibly deeply nested pair and the evaluation derivation of the translated term grows accordingly. Generally, $\Lambda l.V$ as a closed value is in the form of $\Lambda l. \Lambda l_1. \cdots \Lambda l_n. V_n$ where $V_n=\lamL{\Loc}{x}{M_0}$ or $V_n=\Lambda\alpha.W$ due to the syntactic restriction on $V$. When $n$ is zero, the translated term is a flat pair. When $n$ is greater than zero, the translated pair looks more or less like a full binary tree of height $n$. By the second induction, the proof constructs an evaluation derivation of the same as the structure of the translated term of this case. 
	
	In (Lapp), we prove that if $\evalRPC{L[b]}{a}{V_0\subst{b}{l}}$ then $\evalRPC{\mono{L[b]}}{a}{\mono{V_0\subst{b}{l}}}$ under the assumption that $\evalRPC{L}{a}{\Lambda l. V_0}$. The proof involves two subcases, one with $b=\client$ and the other with $b=\server$, both of which can be proved with no difficulty. 

	The proof of Lemma \ref{lemma:substitutionmono} and Lemma \ref{lemma:typesubstitutiontermmono} is available in the appendix.

\begin{lemma}[Substitution under  monomorphization] $\mono{M}\subst{\mono{W}}{x} = \mono{M\subst{W}{x}}$. 
\label{lemma:substitutionmono}
\end{lemma}

\begin{lemma}[Type substitution over term under  monomorphization] $\mono{V}\subst{\mono{B}}{\alpha} = \mono{V\subst{B}{\alpha}}$. 
\label{lemma:typesubstitutiontermmono}
\end{lemma}

\begin{theorem}[Semantic correctness of  monomorphization] If \ $\typing{\emptyset}{a}{M}{A}$ and $\evalRPC{M}{a}{V}$, then $\evalRPC{\mono{M}}{a}{\mono{V}}$.
\label{thm:semanticcorrectness}
\end{theorem}
\begin{proof}
We prove this theorem by induction on the height of the evaluation derivation. Base cases involve (Abs) and (Labs), and inductive cases use (App) and (Lapp).

{\bf (Abs):}
	$\evalRPC{\lamL{b}{x}{M_0}}{a}{\lamL{b}{x}{M_0}}$ where  $M=V=\lamL{b}{x}{M_0}$. Therefore, $\evalRPC{ \mono{M} }{a}{ \mono{V} }$, which proves this case.
	
{\bf (Tabs):}
	$\evalRPC{\Lambda\alpha.V_0}{a}{\Lambda\alpha.V_0}$ where $M=V=\Lambda\alpha.V_0$. Therefore, $\evalRPC{ \mono{M} }{a}{ \mono{V} }$, which proves this case.

{\bf (Labs):}
	$\evalRPC{\Lambda l.V_0}{a}{\Lambda l.V_0}$ where $M=V=\Lambda l.V_0$. By a careful analysis, it turns out that $V_0$ is written as $\Lambda l_1. \cdots \Lambda l_n. V_n$ for $n\geq 0$. Let us name $\Lambda l_{i+1}.V_{i+1}$ as $V_i$ for $0\leq i\leq n-1$, and $V_n$ is either $\lamL{\Loc}{x}{M_0}$ or $\Lambda\alpha.W$. Note that all $V_i$s and $V_n$ must not be a variable since the empty typing environment is in the hypothesis of the typing derivation over $M$. Hence, we need to prove
\[
	\evalRPC{ \mono{\Lambda l. \Lambda l_1. \cdots \Lambda l_n. V_n} }{a}{ \mono{\Lambda l. \Lambda l_1. \cdots \Lambda l_n. V_n} }
\]
	
	We prove this case by the second induction on the number of the nested location abstractions, $n$. When $n=0$, consider $V_0=\lamL{\Loc}{x}{M_0}$.
	\[
	\mono{M}
	= \mono{V}
	= \mono{\Lambda l.\lamL{\Loc}{x}{M_0}}
	= ( \mono{ (\lamL{\Loc}{x}{M_0})\subst{\client}{l} }, \mono{ (\lamL{\Loc}{x}{M_0})\subst{\server}{l} }) 
	\]
	which is translated into $ (  \lamL{\Loc\subst{\client}{l}} {x}{ \mono{M_0\subst{\client}{l}} }, 
		\lamL{\Loc\subst{\server}{l}}{x}{ \mono{M_0\subst{\server}{l}} } )$. This translation is well-defined in either case over $\Loc$. If $\Loc$ is some $b$ then $\Loc\subst{a}{l}$ is $\Loc$, which is $b$. If $\Loc$ is a location variable, then it must be $l$ since $M$ is a top-level closed term. So, $\Loc\subst{a}{l}$ is $a$. 
		
	By (Abs), (1):$\evalRPC{\lamL{\Loc\subst{\client}{l}} {x}{ \mono{M_0\subst{\client}{l}} }}{a}{\lamL{\Loc\subst{\client}{l}} {x}{ \mono{M_0\subst{\client}{l}} }}$.
	
	By (Abs), (2):$\evalRPC{\lamL{\Loc\subst{\server}{l}}{x}{ \mono{M_0\subst{\server}{l}} }}{a}{\lamL{\Loc\subst{\server}{l}}{x}{ \mono{M_0\subst{\server}{l}} }}$.
	
	By (Pair) in the evaluation rule for {\rpc} with (1) and (2), $\evalRPC{M}{a}{V}$, which proves this base case with $V_0=\lamL{\Loc}{x}{M_0}$ of the second induction.  
	
	Also when $n=0$, consider the other base case with $V_0=\Lambda\alpha.W$.
	\[
	\mono{M}
	= \mono{V}
	= \mono{\Lambda l.\Lambda\alpha.W}
	= ( \mono{ (\Lambda\alpha.W)\subst{\client}{l} }, \mono{ (\Lambda\alpha.W)\subst{\server}{l} }) 
	\]
	which is translated into $( \Lambda\alpha. \mono{ W\subst{\client}{l} } ,  \Lambda\alpha.\mono{W\subst{\server}{l} })$.
	
	By (Tabs), (1)':$\evalRPC{ \Lambda\alpha. \mono{ W\subst{\client}{l} }  } {a} { \Lambda\alpha. \mono{ W\subst{\client}{l} }  }$.
	
	By (Tabs), (2)':$\evalRPC{ \Lambda\alpha.\mono{W\subst{\server}{l} } } {a} { \Lambda\alpha.\mono{W\subst{\server}{l} } }$. 
	
	By (Pair) with (1)' and (2)', $\evalRPC{M}{a}{V}$, which proves the remaining base case with $V_0=\Lambda\alpha.W$ of the second induction.  
	
	Now let us prove the inductive case of the second induction, $n>0$. 
	\[
	\mono{M}
	= \mono{V}
	= \mono{\Lambda l.\Lambda l_1.V_1}
	= ( \mono{ (\Lambda l_1.V_1)\subst{\client}{l} }, \mono{ (\Lambda l_1.V_1)\subst{\server}{l} })  
	\]
	which is $( \mono{ \Lambda l_1.(V_1\subst{\client}{l}) }, \mono{ \Lambda l_1.(V_1\subst{\server}{l}) })$ since $l$ must not be equal to $l_1$. Then, by I.H. with $n-1$ of the second induction, 
	\[
	(3):\evalRPC{\mono{ \Lambda l_1.(V_1\subst{\client}{l}) }}{a}{ \mono{ \Lambda l_1.(V_1\subst{\client}{l}) } }
	\ \ \ \ \ 
	(4):\evalRPC{\mono{ \Lambda l_1.(V_1\subst{\server}{l}) }}{a}{ \mono{ \Lambda l_1.(V_1\subst{\server}{l}) } }
	\]
	By (Pair) in the evaluation rule for {\rpc} with (3) and (4), $\evalRPC{ \mono{M} }{a}{ \mono{V} }$, which proves this inductive case of the second induction.  Therefore, we prove this case of (Labs).

{\bf (App):}
	$\evalRPC{L_1 M_1}{a}{V}$, (1):$\evalRPC{L_1}{a}{\lamL{b}{x}{M_0}}$, (2):$\evalRPC{M_1}{a}{W}$, and (3):$\evalRPC{M_0\subst{W}{x}}{b}{V}$ where $M=L_1M_1$.
	By (T-App), (4):$\typing{\emptyset}{a}{L_1}{B\funL{b} A}$ and (5):$\typing{\emptyset}{a}{M_1}{B}$.
	
	By I.H. with (1) and (4), $\evalRPC{\mono{L_1}}{a}{\mono{\lamL{b}{x}{M_0}}}$, which is (6):$\evalRPC{\mono{L_1}}{a}{\lamL{b}{x}{\mono{M_0}}}$. 
	
	By I.H. with (2) and (5), (7):$\evalRPC{\mono{M_1}}{a}{W}$.
	
	By the type soundness theorem with (1) and (4), (8):$\typing{\emptyset}{a}{\lamL{b}{x}{M_0}}{B\funL{b}A}$. 
	
	By the type soundness theorem with (2) and (5), (9):$\typing{\emptyset}{a}{W}{B}$.
	
	By the lemma (Value substitution) with (8) and (9), (10):$\typing{\emptyset}{b}{M_0\subst{W}{x}}{A}$.
	
	By I.H. with (3) and (10), (11):$\evalRPC{ \mono{M_0\subst{W}{x}} }{b}{ \mono{V} }$.
	
	By the lemma (Substitution with the monomorphization translation) with (11), (12):$\evalRPC{ \mono{M_0}\subst{\mono{W}}{x} }{b}{ \mono{V} }$.
	
	By (App) with (6), (7), and (12), $\evalRPC{\mono{L_1} \mono{M_1}}{a}{\mono{V}}$, which proves this case since $\mono{L_1 M_1}=\mono{L_1} \mono{M_1}$.
	
{\bf (Tapp):}
	$\evalRPC{M_1[B]}{a}{V_0\subst{B}{\alpha}}$, and (1):$\evalRPC{M_1}{a}{\Lambda\alpha.V_0}$ where $M=M_1[B]$ and $V=V_0\subst{B}{\alpha}$. 
	
	By (T-Tapp), $\typing{\emptyset}{a}{M_1[B]}{A_0\subst{B}{\alpha}}$, (2):$\typing{\emptyset}{a}{M_1}{\forall\alpha.A_0}$ where $A=A_0\subst{B}{\alpha}$.
	
	By I.H. with (1) and (2), $\evalRPC{ \mono{M_1} }{a}{ \mono{\Lambda\alpha.V_0} }$, which is (3):$\evalRPC{ \mono{M_1} }{a}{ \Lambda\alpha.\mono{V_0} }$ by the def. of the translation. 
	
	By (Tapp) with (3), $\evalRPC{ \mono{M_1}[ \mono{B} ] }{a}{ \mono{V_0}\subst{ \mono{B} }{\alpha}}$. 
	
	By the lemma \ref{lemma:typesubstitutiontermmono}, $\evalRPC { \mono{M_1}[ \mono{B} ] }{a}{ \mono{V_0\subst{B}{\alpha}} }$, which proves this case.

{\bf (Lapp):}
	$\evalRPC{L[b]}{a}{V_0\subst{b}{l}}$ and (1):$\evalRPC{L}{a}{\Lambda l.V_0}$ where $M=L[b]$ and $V=V_0\subst{b}{l}$. By (T-Lapp), $\typing{\emptyset}{a}{L[b]}{A_0\subst{b}{l}}$ and (2):$\typing{\emptyset}{L}{a}{\forall l.A_0}$ where $A=A_0\subst{b}{l}$.
	
	By applying I.H. to (1) and (2), $\evalRPC{ \mono{L} }{a}{ \mono{\Lambda l.V_0} }$. By the definition of the translation, (3):$\evalRPC{ \mono{L} }{a}{ ( \mono{V_0\subst{\client}{l}}, \mono{V_0\subst{\server}{l}} ) }$.
	
	Since $\mono{ L[b] }$ is $\pi_1(\mono{L})$ if $b=\client$, and it is $\pi_2(\mono{L})$ if $b=\server$, we prove two sub cases. Suppose $b=\client$. By (Pair) in the evaluation rule for {\rpc} with (3), $\evalRPC{\pi_1(\mono{L})}{a}{V_0\subst{\client}{l}}$, which proves one case of (Lapp) when $b=\client$. Now suppose $b=\server$. By (Pair) with (3) again, $\evalRPC{\pi_2(\mono{L})}{a}{V_0\subst{\server}{l}}$, which proves the other case of (Lapp).
	
\end{proof}

\subsection{On putting the polymorphic RPC calculus into practice}
\label{sec:howtoputintopractice}

	Now we have finished presenting our theory of the polymorphic RPC calculus. A strong point is that this is a new RPC calculus supporting polymorphic locations that conservatively extends the typed RPC calculus. % based on the the untyped RPC calculus. 

	To put this theory into practice, however, there is a weak point to address. Although our theory  clearly accounts for what is the polymorphic RPC calculus, the monomorphization translation, on which the theory depends, can potentially lead to code explosion, as was explained in Section \ref{subsubsec:monomorphizationexamples}. The reason to perform the monomorphization translation is that after it, one can determine statically if every lambda abstraction is a local or remote procedure call. We call this a {\it static} approach to the polymorphic RPC calculus. 
	
	An alternative way is what we call a {\it dynamic} approach that allows one to determine dynamically if a given lambda application is a local or remote procedure call in runtime after the slicing  compilation. However, in the existing client-server calculi, given a function $V$ and an argument $W$, three application terms are local procedure call, $V(W)$, a remote procedure call from the client to the server, $\req{V}{W}$, and the other remote procedure call for the reverse direction, $\call{V}{W}$, and they can be used only where the caller and the callee locations are statically resolved. They are not equipped with such dynamic operation on locations in runtime necessary for lambda applications that are not certain whether they are remote procedure calls. For example, the cross-tier composition function ($\circ$) in Section \ref{subsec:polyrpcexample} gets a typing for a subterm $g \ x$ as 
	\[
	\typing{\{\alpha \beta \gamma, \ l_1 l_2, \ f:\beta\funL{l_2}\gamma, g:\alpha\funL{l_1}\beta , x:\alpha\}}{l_2}{g \ x}{\beta}
	\] 
	where $l_1$ is where to run the function $g$ and $l_2$ is where the application is. It cannot be statically determined if $g \ x$ is a remote procedure call or not. 
	
	For the dynamic approach to the polymorphic RPC calculus, we need to introduce a new application term that allows such dynamic location checking in runtime as this:
	\[ \gen{Loc'}{V}{W}
	\]
where $Loc'$ is the (callee) location of where to run the function $V$. Note that the caller location, say, $Loc$, is where this application term runs, which is easy to obtain in runtime.
	Whenever a lambda application in {\polyrpc}  poses uncertain caller or callee locations, it can be compiled into this new dynamic application term in a client-server calculus, which we might call a {\it polymorphic CS calculus}, {\polycs}.

	The dynamic semantics for $\gen{Loc'}{f}{arg}$ at $Loc$ can be defined as in the following table. 
\begin{center}
\begin{tabular}{| c | c | c | c | } \hline
Caller($Loc$), Callee($Loc'$) 
%& $\lambda^{typed}_{rpc}$ 
& $\gen{Loc'}{f}{arg}$ & procedure call (flow) \\ \hline \hline
$Loc = Loc'=a$ 
% & $L \ M$ 
& $f(arg)$ & Local  (a $\rightarrow$a)\\ 
$Loc = \client$ and $Loc'= \server$ 
% & $L \ M$ 
& $req(f, arg)$ & Remote  (\client $\rightarrow$ \server)\\
$Loc = \server$ and $Loc'= \client$ 
% & $L \ M$ 
& $call(f, arg)$ & Remote (\server $\rightarrow$ \client)\\  \hline
\end{tabular}
\end{center}

\ \\ 
	The first column of the table shows three cases of dynamic checking on caller and callee locations.
	Note that the caller and the callee locations become monomorphic in runtime.
	When the caller and callee locations match to one of the three cases, the semantics for the generic application term is defined by the semantics for the matched specific application term in the second column.  

	%We believe that this design of the generic application term is morally equivalent to what Links currently does. 

	Now a slicing compilation may simply compile a lambda application $L \ N$ in {\polyrpc}  into a generic application term in {\polycs}:
	\[
		\mathcal{C}\mono{L \ N}_{\tyenv,Loc,B} = 
			\gen{Loc'}
					{ \mathcal{C}\mono{L}_{\tyenv,Loc,A\funL{Loc'},B} }
					{\mathcal{C}\mono{N}_{\tyenv,Loc,A}}
	\]	
	where $\mathcal{C}{\mono{M}}_{\tyenv,Loc,A}$ denotes a compilation of a term $M$ that has type $A$ at location $Loc$ under a type environment $\tyenv$. 
	Then the generic application term will do local or remote procedure calls under the dynamic semantics in {\polycs}, as explained previously.
	For example, in the cross-tier composition function, the subterm $g \ x$ under the typing explained above is compiled into $\gen{l_1}{g}{x}$ at the location $l_2$. 

	Note that whenever the caller and callee locations are determined statically for comparison in compile-time, it is possible to optimize the compilation of lambda applications in {\polyrpc}  so as to generate more specific application terms in {\polycs}  for the purpose of avoiding dynamic location checking. 

	The dynamic approach could be further optimized in terms of the reduction of dynamically passing locations and checking them if we could make use of the monomorphization translation in the static approach  without risking code explosion. Some polymorphic location functions, say, with only one location abstraction such as $\Lambda l.\Lambda\alpha.\lamL{l}{x}{x}$, had better be monomorphized into a client version and a server version. This is because the polymorphic functions will be placed anyway both in the client and in the server  after the slicing compilation. 
	
	The design of a {\it selective} monomorphization translation would be also useful in making the polymorphic RPC calculus into practice. A basic idea is to classify location variables by {\it kind} indicating if the location variables are static or dynamic. That is, we may replace $\Lambda l.V$ by 
	\[ \Lambda l:k.V \ \mbox{where the kind $k$ is either static or dynamic.} \]
	For example, two location abstractions in the cross-tier composition function may get kind information as:
	\[
		\Lambda l_1:static. \ \Lambda l_2:dynamic. \ 
			\Lambda \alpha. \Lambda \beta. \Lambda \gamma.
			\lamL{l_2}{f}{\lamL{l_2}{g}{\lamL{l_2}{x}
				{f \ (g \ x)}}}.
	\]
	
	Then the selective monomorphization translation would be only applied to the polymorphic locations of the static kind. 
	\[ \mono{\Lambda l:k. V}_{\forall l:k.A} = 
	\begin{cases} 
	(\mono{V\subst{\client}{l}}_{A\subst{\client}{l}}, \mono{V\subst{\server}{l}}_{A\subst{\server}{l}})  & \texttt{if} \ k=\mbox{static}\\
	\Lambda l:k.\mono{ V }_{A} & \texttt{if} \ k = \mbox{dynamic}
	\end{cases}
	\]
where $\mono{M}_A$ denotes the selective monomorphization translation for a term $M$ of type $A$.
	When the kind $k$ is static, it is the same as the monomorphization translation. Otherwise, it leaves the location abstraction and translates the body. 
	The remaining dynamic kinded locations after the selective translation would be supported by the polymorphic CS calculus.
	For example, the location-kinded cross-tier composition function is selectively monomorphized into one instantiated with $\client$ for $l_1$ and the other with $\server$ for $l_1$. A benefit is that locations for $l_1$ can be statically resolved incurring no burden in run-time due to dynamically passing them. 
		
	Now a question on the selective monomorphization is how to decide if given location variables are static or dynamic. Generally, there are two options. Programmers  may assign location variables appropriate kinds manually, or some static analysis may do this automatically.

	While the dynamic approach thus solves the potential code explosion problem of the static approach,  it seems to require a {\it type-passing semantics} where type information is formed and passed to polymorphic functions during run-time.  But only the generic application terms depend on location information. Nothing else depends on location nor type information. So, we would like to have a {\it type-erasure semantics} where types will ultimately be erased and the same term represents different instantiations of polymorphic functions at run-time, resulting in no run-time cost. 
	To retain location information in the type-erasure semantics, we need to introduce {\it location representations} that can carry run-time location information about location variables  and can be type-safely inspected. 
	 For this purpose, a generalized algebraic data type (GADT \cite{Xi:10.1145/640128.604150} or first-class phantom type \cite{cheney:tr2003-1901}) can be used. For example $Location \ \alpha$ is such a GADT that has two constructors $\hat{\client}$ of type $Location \ Client$ and $\hat{\server}$ of type $Location \ Server$ where $Client$ and $Server$ are some types. Then every generic application term $\gen{Loc}{M}{N}$ in the type-passing semantics can be implemented by a variant as $\gen{L}{M}{N}$ in the type-erasure semantics where $L$ is a term of $Location \ A$. The type $A$ will be one of $\alpha$, $Client$, and $Server$ when $Loc$ is one of $l$, $\client$, and $\server$, respectively. The variant generic application terms can be implemented as a case expression, $\mbox{case} \ L \ \mbox{of} \ \{ \ \hat{\client}\rightarrow M; \ \hat{\server}\rightarrow N \ \}$. This shows that using GADT to encode location information fits well for our purpose.

%	While the dynamic approach thus solves the potential code explosion problem of the static approach, it could cause some overhead due to a {\it location-passing} interpretation of location polymorphism in which locations are constructed, passed, and analyzed as data at run time. So, we would like to use the dynamic approach in a {\it location-erasure} interpretation of location polymorphism. To retain location information in the location-erasure interpretation, we need to introduce {\it location representations} that can carry run-time location information about location variables  and can be type-safely inspected.  A generalized algebraic data type (GADT \cite{Xi:10.1145/640128.604150}, also first-class phantom type \cite{cheney:tr2003-1901}) can be used for this purpose. For example $Location \ \alpha$ is such a GADT that has two constructors $\hat{\client}$ of type $Location \ Client$ and $\hat{\server}$ of type $Location \ Server$ where $Client$ and $Server$ are some types. Then $\hat{\client}$ is the only term for the type $Location \ Client$, and so is $\hat{\server}$ for the type $Location \ Server$. With this encoding, we would refer to a type-theoretic framework that supports intensional polymorphism in type-erasure semantics \cite{Cray:1998:10.1145/291251.289459}.
	
	In summary, the polymorphic RPC calculus can be supported either by a static approach or a dynamic approach. In the static approach, no dynamic location checking is required in runtime but there is some potential code explosion problem. In the dynamic approach, no code explosion problem will happen but this advantage comes at the cost of dynamically passing locations and checking them in runtime. 
	
	The following table summarizes the status. 

\begin{center}
\begin{tabular}{|c|c|c|c|c|} \hline
 RPC calculi    & Dynamic check  & Location type    & Slicing \\ \hline
Untyped RPC          & always             & (untyped)        & yes \\
Typed RPC           & never              & mono. loc. & yes \\
Poly. RPC (Static)  & never  (code size $\uparrow$) & poly. loc. & yes \\
Poly. RPC (Dynamic) & only for poly loc. & poly. loc. & yes \\ \hline 
\end{tabular}
\end{center}
%\ \\
%where (*) means the dynamic location checking on local and remote procedure calls. 

\section{Related work and Discussion}
\label{sec:relatedwork}

\paragraph{Polymorphic locations}
	There are only  a few publications that are relevant to the notion of polymorphic locations. ML5 supports what they call {\it world polymorphism} \cite{Murphy:2007:TDP:1793574.1793585,Murphy:2008:MTM:1467784}.  Let us first explain ML5 briefly. A key construct for remote evaluation is $from \ L \ get \ M$ where $L$ is supposed to evaluate to a world and $M$ is an expression at the world.
	For example, here is an example program that involves two worlds, home and server:

\ \ \ \ \ 
\ \ \ \ \ 
\begin{minipage}{0.9\textwidth}
\begin{lstlisting}[escapechar=\%,language=lisp]
extern bytecode world server
extern val server : server addr @ home
extern val version : unit -> string @ server
extern val alert : string -> unit @ home

fun showversion() = 
      let val s = from server
                  get version ()
      in  alert [Server's version is: [s]]
      end

do showversion()                   
\end{lstlisting}
\end{minipage}
where a client at world named $home$ invokes a client function $showversion$ and, subsequently, it invokes a server function $version$ at world named $server$ to retrieve a version string and to display it. 

	In ML5, a function is said to be {\it valid}, which means it can be used at any world, when it does not access any local resources. An example is $map$, which roughly has type $\Box w. (A\rightarrow B)\rightarrow([A] \rightarrow [B])$ using the box modal construct\footnote{In fact, it is {\it sham}, a variant of the box modal construct.}  over a bound world variable $w$. Using $map$, we can write as \[ from \ server \ get \ (map \ (fn \ x => x +1) \ [1,2,3]). \] 
	
	Although the use of quantified modal construct may look like the use of the universal quantifier over a bound location variable $\forall l$,  they are subtly different from each other. Every world variable quantified by a box modal construct is instantiated with some current world. In the example, $w$ becomes instantiated with $server$ in the server where $map$ is used. The polymorphic RPC calculus expresses location polymorphism by parametric polymorphism. In the example, the client would select a server version of $map$ by $map[\server]$ in the client, and would invoke the selected server map function. In this respect, while ML5 can be called modally quantified polymorphism, the polymorphic RPC calculus is based on parametric polymorphism. At least, our approach is believed to be more suitable for the Links programming language. It would be interesting to investigate a formal comparison between world polymorphism and location polymorphism.
	
	For implementation of polymorphism, ML5 represents  worlds at run-time to specialize the representation of values given its world while, in (the static approach to) the polymorphic RPC calculus, we compile all location abstractions and applications  away during compile-time to know communication flows  statically. Here, worlds are treated as concrete addresses of distributed computers while locations are regarded as abstract categories of distributed computers, i.e., client and server. That is why the ML5 example imports the server address as an external value to use it in the client. 
	However, both of ML5 and the polymorphic RPC calculus can take an advantage of each other's approach. 
	On the one hand, ML5 could enjoy something like our monomorphization translation, for example, as an optimization. On the other hand, the polymorphic RPC calculus could represent locations at run-time after introducing a dynamic operation on locations as in the dynamic approach.
	
%	To write distributed programs we also need dynamic tokens with which to
%refer to worlds. A token for the world w has type w addr, and can be thought
%of as the address of that world. A world can compute its own address with the
%localhost() expression, whose typing rule appears in Figure 1.
%	Typically, a
%program also expects to know the addresses of other worlds when it begins and
%imports them with extern val. We use an address by traveling to the world it
%indicates, using the get construct. 
%
%	The modality (pronounced “shamrock”) is the internalization of the validity judg-
%ment as a type. Operationally, a value of type A is an encapsulated value of type A that
%makes sense at any world. (Again the hypothetical world can appear in the type; when
%it does we write %ω 
%A.) Its introduction rule is the same as the rule for 2 in the proof
%theory, except for the possible appearance of the bound world variable. (In the dynamic
%semantics (Section 3.5.2) we will require the body of the shamrock to be a value.) The
%elimination rule E binds a valid variable in the body of the letsham.
%Like the other modalities, A is mobile for any A.
%
%	In the ML5/pgh internals we will perform a type representation transformation (Sec-
%tions 5.4.5 and 5.4.7) to make sure that whenever we marshal a value, we have a repre-
%sentation of that value’s type. Because the modal type system assigns both types and
%worlds to values, we additionally represent worlds at run-time. This is important be-
%cause it allows us to specialize the representation of some value given its world. (P.127: Dom node example)
	
%%%
	In Eliom \cite{Radanne2017,Radanne:2018:TWP:3184558.3185953}, there are a few features to discuss for comparison. First, it supports a macro feature called {\it shared sections}, which makes it possible to write code for the client and for the server at the same time. It is reported that this technique is pervasively used in Eliom to expose implementations that can be used either on the  client or on the server with similar semantics, in a very concise way. But this is a purely syntactic transformation, which is implemented simply by duplicating the code before type checking.
	
	Second, for integration with the OCaml language, Eliom introduces a third location called $base$. Code located on base can be used both on the client and on the server. This feature allows Eliom to be integrated with the OCaml ecosystem smoothly.  The type checker reports an error if Eliom programs put on the base location other than OCaml constructs. When a multi-tier programming language is designed to be based on some existing programming language, this feature will be useful. 
	
	Third, Eliom allows the same module to mix declarations from multiple locations. Such modules are called $mixed$. This allows programmers to group together declarations that are semantically related, regardless of client-server boundaries.  The module type checker enforces an important constraint that as you go down inside sub modules, locations should be properly included. A client module cannot contain server declarations and conversely, but mixed modules can contain everything. It would be interesting to extend the polymorphic RPC calculus with multi-tier ML modules like the ones provided by Eliom. 	% Rich module systems such as ML’s are notoriously difficult to formalize and implement. Adding tierless elements to the mix certainly would not make the situation simpler, as the authors of Eliom stated.

\paragraph{Location inference}
	
	The polymorphic RPC calculus needs a location inference method in practice. Without it, programmers would be burdened because they have to annotate locations on applications as well as on lambda abstractions. It is desirable to have an automatic location inference method.  
	
	In the design of a location inference algorithm, an issue is how location abstractions and applications are written in polymorphic RPC programs. A solution is to extend type inference methods for the System F calculus since  the inference problem for type abstractions and applications is thought to be similar to one for location counterparts in the polymorphic RPC calculus. Although it is well-known to be an undecidable problem in general \cite{Wells1993}, there have been some practical trade-offs including \cite{MILNER1978348,Dunfield:2013:CEB:2500365.2500582,Serrano:2018:GIP:3192366.3192389}. It would be interesting to see whether their methods can be applied to location inferences successfully. 
	
	Besides inference for location constructs, location inference may be designed to insert coercions where there are location conflicts to accept more programs as well-located ones.  For example, $\lamL{\client}{(f:A\funL{\server} B)}{\ \cdots}) \ (\lamL{\client}{x}{\ \cdots})$ is ill-typed but the term can be slightly transformed to be well-typed by inserting an $\eta$-conversion coercion over the argument automatically as: $\lamL{\client}{(f:A\funL{\server} B)}{\ \cdots}) \ (\lamL{\server}{y}{ (\lamL{c}{x}{\ \cdots}) \ y})$.
	
	Regarding ML5, it is reported that a simple extension of Hindley-Milner type inference algorithm is developed. The programmer does not usually need to use the box modality manually, because type inference will automatically generalize declarations to be valid whenever possible \cite{Murphy:2007:TDP:1793574.1793585,Murphy:2008:MTM:1467784}.
	
	Neubauer developed  a constraint-based static analysis that statically infers optimal location annotations for operations in the multi-tier calculus \cite{neubauer2007}.
	
%\paragraph{Monomorphization techniques}
%	doubling translation that simulates each {\bf any} function as a pair of a {\bf pl} function and a {\bf db} function \cite{Cheney:2014:EQR:2543728.2543738}
%	 
%	Polymorphic record compilation and type class compilation. Polymorphism is resolved by selecting an appropriate operation from a dictionary that is provided on instantiating the polymorphism with some monomorphic type. The idea in this paper is more directly affected by the technique used in the work of the two authors. 
%	
%	Our approach is similar, but on polymorphic abstractions, we can prepare a dictinary, which is a pair of a client version and a server version. On instantiation, we have only to select an appropriate version depending on location of interest. 
%	
%	This is an example of ad hoc polymorphism where its term does not 

\paragraph{Multi-tier calculi}

	There have been several multi-tier programming languages being developed. 	
		Links \cite{Cooper:2006:LWP:1777707.1777724} is a multi-tier web programming language that employs the RPC calculus as the foundation for client-server communication. 
	Lambda5 \cite{MurphyVII:2004:SML:1018438.1021865,Murphy:2008:MTM:1467784} is a modally-typed lambda calculus in which modal type systems can control local resources safely in distributed systems. 
	A multi-tier calculus by Neubauer and Thiemann \cite{Neubauer:2005:SPM:1040305.1040324} automatically constructs communications for concurrently running processes employing session types \cite{Gay:1999:TSC:645393.756448} to enforce the integrity of communications. They proposed a series of transformations as compilation to convert a source program into separate programs at different locations determined by the use of primitives that run only at specific locations.
	There are many other multi-tier web  programming languages as follows.
Hop \cite{Serrano2006, Serrano:2012:MPH:2240236.2240253} extending Scheme; Hop.js extending JavaScript \cite{Serrano:2016:GH:3022670.2951916};	Eliom, a multi-tier ML programming language, featuring module systems extended with location annotations \cite{Radanne2017} in the project Ocsigen \cite{Balat2006}; and Ur/Web \cite{Chlipala2015} with a dependently typed system; a multi-tier functional reactive programming framework, ScalaLoci \cite{Weisenburger:2018:DSD:3288538.3276499}, and another interesting framework for Scala \cite{Reynders2014}. 

\section{Conclusion}
\label{sec:conclusion}

	This paper proposed the polymorphic RPC calculus where programmers can write succinct multi-tier programs using polymorphic location constructs and the polymorphic multi-tier programs can be automatically translated into monomorphic multi-tier programs only with location constants amenable to the existing slicing compilation methods for client-server model.  For the polymorphic RPC calculus, we formulated the type system, and proved its type soundness. Also, we designed the monomorphization translation, and we proved its type and semantic correctness for the translation. 
	
	As future work, we want to integrate the polymorphic RPC calculus into the Links programming language \cite{Cooper:2006:LWP:1777707.1777724}, which was once designed based on the untyped RPC calculus.  But today Links offers many modern features of programming languages on a call-by-value variant of System F with row polymorphism, row-based effect types, and implicit subkinding \cite{Lindley:2012:RET:2103786.2103798,hillerstrm_et_al:LIPIcs:2017:7739,Fowler:2019:EAS:3302515.3290341}. Therefore it would be a challenge how the polymorphic RPC type system can be integrated smoothly to work with these features.  
	
	Another interesting direction is to mechanize the semantics of the polymorphic RPC calculus using a proof assistant to substantiate the proofs in this work and to make the mechanized semantics a basis for further development.

%\section*{References}

\bibliography{sbmf2019}

\begin{thebibliography}{28}
\expandafter\ifx\csname natexlab\endcsname\relax\def\natexlab#1{#1}\fi
\providecommand{\url}[1]{\texttt{#1}}
\providecommand{\href}[2]{#2}
\providecommand{\path}[1]{#1}
\providecommand{\DOIprefix}{doi:}
\providecommand{\ArXivprefix}{arXiv:}
\providecommand{\URLprefix}{URL: }
\providecommand{\Pubmedprefix}{pmid:}
\providecommand{\doi}[1]{\href{http://dx.doi.org/#1}{\path{#1}}}
\providecommand{\Pubmed}[1]{\href{pmid:#1}{\path{#1}}}
\providecommand{\bibinfo}[2]{#2}
\ifx\xfnm\relax \def\xfnm[#1]{\unskip,\space#1}\fi
%Type = Inproceedings
\bibitem[{Balat(2006)}]{Balat2006}
\bibinfo{author}{Balat, V.}, \bibinfo{year}{2006}.
\newblock \bibinfo{title}{{Ocsigen: Typing Web Interaction with Objective
  Caml}}, in: \bibinfo{booktitle}{Proceedings of the 2006 Workshop on ML (ML
  '06)}, pp. \bibinfo{pages}{84--94}.
\newblock \DOIprefix\doi{10.1145/1159876.1159889}.
%Type = Techreport
\bibitem[{Cheney and Hinze(2003)}]{cheney:tr2003-1901}
\bibinfo{author}{Cheney, J.}, \bibinfo{author}{Hinze, R.},
  \bibinfo{year}{2003}.
\newblock \bibinfo{title}{First-Class Phantom Types}.
\newblock \bibinfo{type}{Technical Report} \bibinfo{number}{CUCIS TR2003-1901}.
  Cornell University.
%Type = Article
\bibitem[{Chlipala(2015)}]{Chlipala2015}
\bibinfo{author}{Chlipala, A.}, \bibinfo{year}{2015}.
\newblock \bibinfo{title}{{Ur/Web: A Simple Model for Programming the Web}}.
\newblock \bibinfo{journal}{Proceedings of the 42nd Annual ACM SIGPLAN-SIGACT
  Symposium on Principles of Programming Languages - POPL '15} ,
  \bibinfo{pages}{153--165}\URLprefix
  \url{http://dl.acm.org/citation.cfm?doid=2676726.2677004},
  \DOIprefix\doi{10.1145/2676726.2677004}.
%Type = Article
\bibitem[{Choi and Chang(2019)}]{choijfp2019}
\bibinfo{author}{Choi, K.}, \bibinfo{author}{Chang, B.}, \bibinfo{year}{2019}.
\newblock \bibinfo{title}{A theory of {RPC} calculi for client-server model}.
\newblock \bibinfo{journal}{Journal of Functional Programming}
  \bibinfo{volume}{29}, \bibinfo{pages}{e5}.
\newblock \URLprefix \url{https://doi.org/10.1017/S0956796819000029},
  \DOIprefix\doi{10.1017/S0956796819000029}.
%Type = Inproceedings
\bibitem[{Cooper et~al.(2007)Cooper, Lindley, Wadler and
  Yallop}]{Cooper:2006:LWP:1777707.1777724}
\bibinfo{author}{Cooper, E.K.}, \bibinfo{author}{Lindley, S.},
  \bibinfo{author}{Wadler, P.}, \bibinfo{author}{Yallop, J.},
  \bibinfo{year}{2007}.
\newblock \bibinfo{title}{Links: Web programming without tiers}, in:
  \bibinfo{booktitle}{Proceedings of the 5th International Conference on Formal
  Methods for Components and Objects}, \bibinfo{publisher}{Springer-Verlag},
  \bibinfo{address}{Berlin, Heidelberg}. pp. \bibinfo{pages}{266--296}.
\newblock \URLprefix \url{http://dl.acm.org/citation.cfm?id=1777707.1777724}.
%Type = Inproceedings
\bibitem[{Cooper and Wadler(2009)}]{Cooper:2009:RC:1599410.1599439}
\bibinfo{author}{Cooper, E.K.}, \bibinfo{author}{Wadler, P.},
  \bibinfo{year}{2009}.
\newblock \bibinfo{title}{The rpc calculus}, in:
  \bibinfo{booktitle}{Proceedings of the 11th ACM SIGPLAN Conference on
  Principles and Practice of Declarative Programming},
  \bibinfo{publisher}{ACM}, \bibinfo{address}{New York, NY, USA}. pp.
  \bibinfo{pages}{231--242}.
\newblock \URLprefix \url{http://doi.acm.org/10.1145/1599410.1599439},
  \DOIprefix\doi{10.1145/1599410.1599439}.
%Type = Inproceedings
\bibitem[{Dunfield and Krishnaswami(2013)}]{Dunfield:2013:CEB:2500365.2500582}
\bibinfo{author}{Dunfield, J.}, \bibinfo{author}{Krishnaswami, N.R.},
  \bibinfo{year}{2013}.
\newblock \bibinfo{title}{Complete and easy bidirectional typechecking for
  higher-rank polymorphism}, in: \bibinfo{booktitle}{Proceedings of the 18th
  ACM SIGPLAN International Conference on Functional Programming},
  \bibinfo{publisher}{ACM}, \bibinfo{address}{New York, NY, USA}. pp.
  \bibinfo{pages}{429--442}.
\newblock \URLprefix \url{http://doi.acm.org/10.1145/2500365.2500582},
  \DOIprefix\doi{10.1145/2500365.2500582}.
%Type = Article
\bibitem[{Fowler et~al.(2019)Fowler, Lindley, Morris and
  Decova}]{Fowler:2019:EAS:3302515.3290341}
\bibinfo{author}{Fowler, S.}, \bibinfo{author}{Lindley, S.},
  \bibinfo{author}{Morris, J.G.}, \bibinfo{author}{Decova, S.},
  \bibinfo{year}{2019}.
\newblock \bibinfo{title}{Exceptional asynchronous session types: Session types
  without tiers}.
\newblock \bibinfo{journal}{Proc. ACM Program. Lang.} \bibinfo{volume}{3},
  \bibinfo{pages}{28:1--28:29}.
\newblock \URLprefix \url{http://doi.acm.org/10.1145/3290341},
  \DOIprefix\doi{10.1145/3290341}.
%Type = Inproceedings
\bibitem[{Gay and Hole(1999)}]{Gay:1999:TSC:645393.756448}
\bibinfo{author}{Gay, S.J.}, \bibinfo{author}{Hole, M.}, \bibinfo{year}{1999}.
\newblock \bibinfo{title}{Types and subtypes for client-server interactions},
  in: \bibinfo{booktitle}{Proceedings of the 8th European Symposium on
  Programming Languages and Systems}, \bibinfo{publisher}{Springer-Verlag},
  \bibinfo{address}{London, UK, UK}. pp. \bibinfo{pages}{74--90}.
\newblock \URLprefix \url{http://dl.acm.org/citation.cfm?id=645393.756448}.
%Type = Inproceedings
\bibitem[{Harper and Morrisett(1995)}]{Harper:1995:CPU:199448.199475}
\bibinfo{author}{Harper, R.}, \bibinfo{author}{Morrisett, G.},
  \bibinfo{year}{1995}.
\newblock \bibinfo{title}{Compiling polymorphism using intensional type
  analysis}, in: \bibinfo{booktitle}{Proceedings of the 22Nd ACM SIGPLAN-SIGACT
  Symposium on Principles of Programming Languages}, \bibinfo{publisher}{ACM},
  \bibinfo{address}{New York, NY, USA}. pp. \bibinfo{pages}{130--141}.
\newblock \URLprefix \url{http://doi.acm.org/10.1145/199448.199475},
  \DOIprefix\doi{10.1145/199448.199475}.
%Type = Inproceedings
\bibitem[{Hillerstr{\"o}m et~al.(2017)Hillerstr{\"o}m, Lindley, Atkey and
  Sivaramakrishnan}]{hillerstrm_et_al:LIPIcs:2017:7739}
\bibinfo{author}{Hillerstr{\"o}m, D.}, \bibinfo{author}{Lindley, S.},
  \bibinfo{author}{Atkey, R.}, \bibinfo{author}{Sivaramakrishnan, K.C.},
  \bibinfo{year}{2017}.
\newblock \bibinfo{title}{{Continuation Passing Style for Effect Handlers}},
  in: \bibinfo{editor}{Miller, D.} (Ed.), \bibinfo{booktitle}{2nd International
  Conference on Formal Structures for Computation and Deduction (FSCD 2017)},
  \bibinfo{publisher}{Schloss Dagstuhl--Leibniz-Zentrum fuer Informatik},
  \bibinfo{address}{Dagstuhl, Germany}. pp. \bibinfo{pages}{18:1--18:19}.
\newblock \URLprefix \url{http://drops.dagstuhl.de/opus/volltexte/2017/7739},
  \DOIprefix\doi{10.4230/LIPIcs.FSCD.2017.18}.
%Type = Inproceedings
\bibitem[{Lindley and Cheney(2012)}]{Lindley:2012:RET:2103786.2103798}
\bibinfo{author}{Lindley, S.}, \bibinfo{author}{Cheney, J.},
  \bibinfo{year}{2012}.
\newblock \bibinfo{title}{Row-based effect types for database integration}, in:
  \bibinfo{booktitle}{Proceedings of the 8th ACM SIGPLAN Workshop on Types in
  Language Design and Implementation}, \bibinfo{publisher}{ACM},
  \bibinfo{address}{New York, NY, USA}. pp. \bibinfo{pages}{91--102}.
\newblock \URLprefix \url{http://doi.acm.org/10.1145/2103786.2103798},
  \DOIprefix\doi{10.1145/2103786.2103798}.
%Type = Article
\bibitem[{Milner(1978)}]{MILNER1978348}
\bibinfo{author}{Milner, R.}, \bibinfo{year}{1978}.
\newblock \bibinfo{title}{A theory of type polymorphism in programming}.
\newblock \bibinfo{journal}{Journal of Computer and System Sciences}
  \bibinfo{volume}{17}, \bibinfo{pages}{348 -- 375}.
\newblock \URLprefix
  \url{http://www.sciencedirect.com/science/article/pii/0022000078900144},
  \DOIprefix\doi{https://doi.org/10.1016/0022-0000(78)90014-4}.
%Type = Phdthesis
\bibitem[{Murphy(2008)}]{Murphy:2008:MTM:1467784}
\bibinfo{author}{Murphy, VII., T.}, \bibinfo{year}{2008}.
\newblock \bibinfo{title}{Modal Types for Mobile Code}.
\newblock Ph.D. thesis. Carnegie Mellon University.
  \bibinfo{address}{Pittsburgh, PA, USA}.
\newblock \bibinfo{note}{AAI3314655}.
%Type = Inproceedings
\bibitem[{Murphy et~al.(2008)Murphy, Crary and
  Harper}]{Murphy:2007:TDP:1793574.1793585}
\bibinfo{author}{Murphy, VII., T.}, \bibinfo{author}{Crary, K.},
  \bibinfo{author}{Harper, R.}, \bibinfo{year}{2008}.
\newblock \bibinfo{title}{Type-safe distributed programming with ml5}, in:
  \bibinfo{booktitle}{Proceedings of the 3rd Conference on Trustworthy Global
  Computing}, \bibinfo{publisher}{Springer-Verlag}, \bibinfo{address}{Berlin,
  Heidelberg}. pp. \bibinfo{pages}{108--123}.
\newblock \URLprefix \url{http://dl.acm.org/citation.cfm?id=1793574.1793585}.
%Type = Inproceedings
\bibitem[{Murphy~VII et~al.(2004)Murphy~VII, Crary, Harper and
  Pfenning}]{MurphyVII:2004:SML:1018438.1021865}
\bibinfo{author}{Murphy~VII, T.}, \bibinfo{author}{Crary, K.},
  \bibinfo{author}{Harper, R.}, \bibinfo{author}{Pfenning, F.},
  \bibinfo{year}{2004}.
\newblock \bibinfo{title}{A symmetric modal lambda calculus for distributed
  computing}, in: \bibinfo{booktitle}{Proceedings of the 19th Annual IEEE
  Symposium on Logic in Computer Science}, \bibinfo{publisher}{IEEE Computer
  Society}, \bibinfo{address}{Washington, DC, USA}. pp.
  \bibinfo{pages}{286--295}.
\newblock \URLprefix \url{http://dx.doi.org/10.1109/LICS.2004.7},
  \DOIprefix\doi{10.1109/LICS.2004.7}.
%Type = Phdthesis
\bibitem[{Neubauer(2007)}]{neubauer2007}
\bibinfo{author}{Neubauer, M.}, \bibinfo{year}{2007}.
\newblock \bibinfo{title}{{Multi-tier programming}}.
\newblock Ph.D. thesis. Universität Freiburg.
%Type = Inproceedings
\bibitem[{Neubauer and Thiemann(2005)}]{Neubauer:2005:SPM:1040305.1040324}
\bibinfo{author}{Neubauer, M.}, \bibinfo{author}{Thiemann, P.},
  \bibinfo{year}{2005}.
\newblock \bibinfo{title}{From sequential programs to multi-tier applications
  by program transformation}, in: \bibinfo{booktitle}{Proceedings of the 32Nd
  ACM SIGPLAN-SIGACT Symposium on Principles of Programming Languages},
  \bibinfo{publisher}{ACM}, \bibinfo{address}{New York, NY, USA}. pp.
  \bibinfo{pages}{221--232}.
\newblock \URLprefix \url{http://doi.acm.org/10.1145/1040305.1040324},
  \DOIprefix\doi{10.1145/1040305.1040324}.
%Type = Phdthesis
\bibitem[{Radanne(2017)}]{Radanne2017}
\bibinfo{author}{Radanne, G.}, \bibinfo{year}{2017}.
\newblock \bibinfo{title}{{Tierless Web programming in ML}}.
\newblock Ph.D. thesis. University Paris Diderot.
%Type = Inproceedings
\bibitem[{Radanne and Vouillon(2018)}]{Radanne:2018:TWP:3184558.3185953}
\bibinfo{author}{Radanne, G.}, \bibinfo{author}{Vouillon, J.},
  \bibinfo{year}{2018}.
\newblock \bibinfo{title}{Tierless web programming in the large}, in:
  \bibinfo{booktitle}{Companion Proceedings of the The Web Conference 2018},
  \bibinfo{publisher}{International World Wide Web Conferences Steering
  Committee}, \bibinfo{address}{Republic and Canton of Geneva, Switzerland}.
  pp. \bibinfo{pages}{681--689}.
\newblock \URLprefix \url{https://doi.org/10.1145/3184558.3185953},
  \DOIprefix\doi{10.1145/3184558.3185953}.
%Type = Inproceedings
\bibitem[{Reynders et~al.(2014)Reynders, Devriese and Piessens}]{Reynders2014}
\bibinfo{author}{Reynders, B.}, \bibinfo{author}{Devriese, D.},
  \bibinfo{author}{Piessens, F.}, \bibinfo{year}{2014}.
\newblock \bibinfo{title}{{Multi-Tier Functional Reactive Programming for the
  Web}}, in: \bibinfo{booktitle}{Proceedings of the 2014 ACM International
  Symposium on New Ideas, New Paradigms, and Reflections on Programming {\&}
  Software - Onward! '14}, pp. \bibinfo{pages}{55--68}.
\newblock \URLprefix \url{http://dl.acm.org/citation.cfm?doid=2661136.2661140},
  \DOIprefix\doi{10.1145/2661136.2661140}.
%Type = Inproceedings
\bibitem[{Serrano et~al.(2018)Serrano, Hage, Vytiniotis and
  Peyton~Jones}]{Serrano:2018:GIP:3192366.3192389}
\bibinfo{author}{Serrano, A.}, \bibinfo{author}{Hage, J.},
  \bibinfo{author}{Vytiniotis, D.}, \bibinfo{author}{Peyton~Jones, S.},
  \bibinfo{year}{2018}.
\newblock \bibinfo{title}{Guarded impredicative polymorphism}, in:
  \bibinfo{booktitle}{Proceedings of the 39th ACM SIGPLAN Conference on
  Programming Language Design and Implementation}, \bibinfo{publisher}{ACM},
  \bibinfo{address}{New York, NY, USA}. pp. \bibinfo{pages}{783--796}.
\newblock \URLprefix \url{http://doi.acm.org/10.1145/3192366.3192389},
  \DOIprefix\doi{10.1145/3192366.3192389}.
%Type = Article
\bibitem[{Serrano and Berry(2012)}]{Serrano:2012:MPH:2240236.2240253}
\bibinfo{author}{Serrano, M.}, \bibinfo{author}{Berry, G.},
  \bibinfo{year}{2012}.
\newblock \bibinfo{title}{{Multitier Programming in Hop}}.
\newblock \bibinfo{journal}{Commun. ACM} \bibinfo{volume}{55},
  \bibinfo{pages}{53--59}.
\newblock \URLprefix \url{http://doi.acm.org/10.1145/2240236.2240253},
  \DOIprefix\doi{10.1145/2240236.2240253}.
%Type = Inproceedings
\bibitem[{Serrano et~al.(2006)Serrano, Gallesio and Loitsch}]{Serrano2006}
\bibinfo{author}{Serrano, M.}, \bibinfo{author}{Gallesio, E.},
  \bibinfo{author}{Loitsch, F.}, \bibinfo{year}{2006}.
\newblock \bibinfo{title}{{Hop: a Language for Programming the Web 2.0}}, in:
  \bibinfo{booktitle}{Proceedings of the 1st Dynamic Languages Symposium},
  \bibinfo{address}{Portland, OR, USA}. pp. \bibinfo{pages}{975--985}.
\newblock \DOIprefix\doi{10.1145/1176617.1176756}.
%Type = Article
\bibitem[{Serrano and Prunet(2016)}]{Serrano:2016:GH:3022670.2951916}
\bibinfo{author}{Serrano, M.}, \bibinfo{author}{Prunet, V.},
  \bibinfo{year}{2016}.
\newblock \bibinfo{title}{{A Glimpse of Hopjs}}.
\newblock \bibinfo{journal}{SIGPLAN Not.} \bibinfo{volume}{51},
  \bibinfo{pages}{180--192}.
\newblock \URLprefix \url{http://doi.acm.org/10.1145/3022670.2951916},
  \DOIprefix\doi{10.1145/3022670.2951916}.
%Type = Article
\bibitem[{Weisenburger et~al.(2018)Weisenburger, K\"{o}hler and
  Salvaneschi}]{Weisenburger:2018:DSD:3288538.3276499}
\bibinfo{author}{Weisenburger, P.}, \bibinfo{author}{K\"{o}hler, M.},
  \bibinfo{author}{Salvaneschi, G.}, \bibinfo{year}{2018}.
\newblock \bibinfo{title}{Distributed system development with scalaloci}.
\newblock \bibinfo{journal}{Proc. ACM Program. Lang.} \bibinfo{volume}{2},
  \bibinfo{pages}{129:1--129:30}.
\newblock \URLprefix \url{http://doi.acm.org/10.1145/3276499},
  \DOIprefix\doi{10.1145/3276499}.
%Type = Misc
\bibitem[{Wells(1993)}]{Wells1993}
\bibinfo{author}{Wells, J., B.}, \bibinfo{year}{1993}.
\newblock \bibinfo{title}{Typability and type checking in the second-order
  lambda-calculus are equivalent and undecidable}.
\newblock \URLprefix \url{https://open.bu.edu/handle/2144/1474}.
%Type = Article
\bibitem[{Xi et~al.(2003)Xi, Chen and Chen}]{Xi:10.1145/640128.604150}
\bibinfo{author}{Xi, H.}, \bibinfo{author}{Chen, C.}, \bibinfo{author}{Chen,
  G.}, \bibinfo{year}{2003}.
\newblock \bibinfo{title}{Guarded recursive datatype constructors}.
\newblock \bibinfo{journal}{SIGPLAN Not.} \bibinfo{volume}{38},
  \bibinfo{pages}{224–235}.
\newblock \URLprefix \url{https://doi.org/10.1145/640128.604150},
  \DOIprefix\doi{10.1145/640128.604150}.

\end{thebibliography}

\begin{appendix}

\setcounter{lemma}{0}
\renewcommand{\thelemma}{3.\arabic{lemma}}

\setcounter{subsection}{0}
\renewcommand{\thesubsection}{\Alph{section}.\arabic{subsection}}
     
\section{Proofs of Lemmas in Section 3.1 A polymorphic RPC calculus}
\label{app:proofs_3_1}

\subsection{Proofs  of Lemmas in Section  3.1.1 Type soundness}

 \begin{lemma}[Value relocation]
 % \label{lemma:valueloc} 
 If \ $\typing{\tyenv}{\Loc}{V}{A}$ then $\typing{\tyenv}{\Loc'}{V}{A}$.
 \end{lemma}
 \begin{proof}
 This lemma is proved by induction on the height of the typing derivation tree in the condition. Base cases use (T-Var) and (T-Abs). Let $V$ be $x$ where (T-Var) is used. Then we have $\tyenv(x)=A$. By (T-Var) with $\Loc'$ this time, $\typing{\tyenv}{\Loc'}{x}{A}$. This proves one base case with (T-Var). For the other base case, the lemma is provable similarly. 
 
 Inductive cases involve (T-Tabs) and (T-Labs). Let $V$ be $\Lambda\alpha.V_0$ and $A$ be $\forall\alpha.A_0$ where (T-Tabs) is used. The condition of the lemma gives us the subtree with (1):$\typing{\tyenv,\alpha}{\Loc}{V_0}{A_0}$. 

	By applying the induction hypothesis to (1), we will have (2):$\typing{\tyenv,\alpha}{\Loc'}{V_0}{A_0}$. 

	By applying (T-Tabs) to (2), we can derive the typing derivation tree in the conclusion of the lemma: $\typing{\tyenv}{\Loc'}{\Lambda\alpha.V_0}{\forall\alpha.A_0}$. 

	For the other inductive case, the lemma can be proved in a similar way.
 \end{proof}

 \begin{lemma}[Value substitution] 
 % \label{lemma:valuesubst}
If \ $\typing{\tyenv}{\Loc}{\lamL{\Loc'}{x}{M}}{A\funL{\Loc'}B}$ and $\typing{\tyenv}{\Loc}{V}{A}$ then $\typing{\tyenv}{\Loc'}{M\subst{V}{x}}{B}$. 
 \end{lemma}
 \begin{proof} 
 	By (T-Abs) with the first part of the condition, $\typing{\tyenvExt{x}{A}}{\Loc'}{M}{B}$. Since $x$ is a bound variable, $x$ cannot appear as a free variable in $V$. By the lemma (Value relocation), we have $\typing{\tyenv}{\Loc'}{V}{A}$ from the second part.
 	
 	Then we have only to show a generalized lemma as: 
	if (1):$\typing{\tyenvExt{x}{A}}{\Loc'}{M}{B}$, (2):$\typing{\tyenv}{\Loc'}{V}{A}$, and (3):$x\not\in fv(V)$ then (4):$\typing{\tyenv}{\Loc'}{M\subst{V}{x}}{B}$. 

	We prove the generalized lemma by induction on the height of the derivation tree (1).
	The only base case uses (T-Var) with $M=y$. We do case analysis by $y=x$ and $y\not=x$. In either cases, the generalized lemma is provable immediately.
	
	For inductive cases, let us first consider a case using (T-App) with $M=L\ N$.
	The instance of (1) becomes $\typing{\tyenvExt{x}{A}}{\Loc'}{L \ N}{B}$. 

	By (T-App), we have (5):$\typing{\tyenvExt{x}{A}}{\Loc'}{L}{C \funL{\Loc''} B}$ and (6):$\typing{\tyenvExt{x}{A}}{\Loc'}{N}{C}$. 

	By I.H. with (5), (7):$\typing{\tyenv}{\Loc'}{L\subst{V}{x}}{C \funL{\Loc''} B}$. 

	By I.H. with (6), (8):$\typing{\tyenv}{\Loc'}{N\subst{V}{x}}{C}$.

	By (T-App) with (7) and (8), $\typing{\tyenv}{\Loc'}{(L\subst{V}{x}) \ (N\subst{V}{x})}{B}$ where $(L\subst{V}{x}) \ (N\subst{V}{x})$ is $(L \ N)\subst{V}{x}$. Hence, the case is proved.

	The other inductive cases use (T-Abs), (T-Tabs), (T-Tapp), (T-Labs), and (T-Lapp), which are all provable similarly. 
 \end{proof}
 
  \begin{lemma}[Type substitution] 
 % \label{lemma:typesubst}
If \ $\typing{\tyenv}{\Loc}{\Lambda\alpha.V}{\forall\alpha.A}$ then $\typing{\tyenv}{\Loc}{V\subst{B}{\alpha}}{A\subst{B}{\alpha}}$. 
 \end{lemma}
 \begin{proof}
 	By (T-Tabs) with the first part of the condition, $\typing{\tyenvExtWith{\alpha}}{\Loc}{V}{A}$. Since $\alpha$ is a bound type variable, $\alpha$ cannot occur in $\tyenv$, i.e., $\alpha\not\in ftv(\tyenv)$.
 	
 	We prove a generalized lemma as: if (1):$\typing{\tyenvExtWith{\alpha}}{\Loc}{M}{A}$ and (2):$\alpha\not\in ftv(\tyenv)$ then (3):$\typing{\tyenv}{\Loc}{M\subst{B}{\alpha}}{A\subst{B}{\alpha}}$. 

	This generalized lemma is proved by induction on the height of the derivation tree (1). The base case uses (T-Var) with $M=x$. $M\subst{B}{\alpha}=x\subst{B}{\alpha}=x$. $A\subst{B}{\alpha}$ must be $A$. Otherwise, $\alpha$ occurs in $A$ for $\tyenv(x)=A$, which violates (2).
 	
 	For inductive cases, consider a case using (T-Tapp) with $M=L[C]$ and $A=A_0\subst{C}{\beta}$. (1) becomes (4):$\typing{\tyenvExtWith{\alpha}}{\Loc}{L[C]}{A_0\subst{C}{\beta}}$. 
	By (T-Tapp) with (4), we have (5):$\typing{\tyenvExtWith{\alpha}}{\Loc}{L}{\forall\beta.A_0}$.
	By I.H. with (5), (6):$\typing{\tyenv}{\Loc}{L\subst{B}{\alpha}}{(\forall\beta.A_0)\subst{B}{\alpha}}$. Note that $(\forall\beta.A_0)\subst{B}{\alpha}=\forall\beta.(A_0\subst{B}{\alpha})$ since $\alpha\not=\beta$. 
	By applying (T-Tapp) to (6) with $C\subst{B}{\alpha}$, (7):$\typing{\tyenv}{\Loc}{L\subst{B}{\alpha}[C\subst{B}{\alpha}]}{(A_0\subst{B}{\alpha})\subst{C\subst{B}{\alpha}}{\beta}}$. 
	This proves the inductive case by $\typing{\tyenv}{\Loc}{(L[C])\subst{B}{\alpha}}{(A_0\subst{C}{\beta})\subst{B}{\alpha}}$.
	The other inductive cases use (T-Abs), (T-App), (T-Tabs), (T-Labs), and (T-Lapp), which are provable similarly.
 \end{proof}

 \begin{lemma}[Location substitution]
 % \label{lemma:locsubst} 
 If \ $\typing{\tyenv}{\Loc}{\Lambda l.V}{\forall l.A}$ then $\typing{\tyenv}{\Loc}{V\subst{\Loc'}{l}}{A\subst{\Loc'}{l}}$.
 \end{lemma}
 \begin{proof}
 By (T-Labs) with the first part of the condition, $\typing{\tyenv,l}{\Loc}{V}{A}$. Since $l$ is a bound location variable, $l$ cannot occur in $\tyenv$, i.e., $l\not\in flv(\tyenv)$.
 
 We prove a generalized lemma as: if (1):$\typing{\tyenv,l}{\Loc}{M}{A}$ and (2):$l\not\in flv(\tyenv)$ then (3):$\typing{\tyenv}{\Loc}{M\subst{\Loc'}{l}}{A\subst{\Loc'}{l}}$. In the base case, $M=x$. $M\subst{\Loc'}{l} = x\subst{\Loc'}{l} = x$. $A\subst{\Loc'}{l}=A$ because of $\tyenv(x)=A$ and (2). 
 
 For inductive cases, let us first a case using (T-Lapp) with $M=L[\Loc_0]$ and $A=B\subst{\Loc_0}{l_0}$. The instance of (1) becomes (4):$\typing{\tyenv,l}{\Loc}{L[\Loc_0]}{B\subst{\Loc_0}{l_0}}$. 

 By (T-Lapp) with (4), we have (5):$\typing{\tyenv,l}{\Loc}{L}{\forall l_0. B}$. 

 By I.H. with (5), (6):$\typing{\tyenv}{\Loc}{L\subst{\Loc'}{l}}{(\forall l_0. B)\subst{\Loc'}{l}}$. Since $l\not=l_0$, (7):$(\forall l_0. B)\subst{\Loc'}{l}=\forall l_0. (B\subst{\Loc'}{l})$.
 
 By (T-Lapp) with (6), (7), and $\Loc_0\subst{\Loc'}{l}$, we can derive 
 \[\typing{\tyenv}{\Loc}{(L\subst{\Loc'}{l})[\Loc_0\subst{\Loc'}{l}]}{(B\subst{\Loc'}{l})\subst{\Loc_0\subst{\Loc'}{l}}{l_0}}\]
 which is $\typing{\tyenv}{\Loc}{(L[\Loc_0])\subst{\Loc'}{l}}{(B\subst{\Loc_0}{l_0})\subst{\Loc'}{l}}$.
 
 The other inductive cases use (T-Abs), (T-App), (T-Tabs), (T-Tapp), and (T-Labs), which are proved similarly.
 \end{proof}
 
\section{Proofs  of Lemmas in Section 3.2 A monomorphization translation of the polymorphic RPC calculus}
\label{app:proofs_3_2}

\setcounter{lemma}{4}

\subsection{Proofs  of Lemmas in Section  3.2.2 Type correctness}

\begin{lemma}[Type substitution over type under monomorphization] 
$\mono{A}\subst{\mono{B}}{\alpha}$  $=$ $\mono{A\subst{B}{\alpha}}$.
%\label{lemma:typesubstitutionmono}
\end{lemma}
 \begin{proof} We prove this lemma by the structural induction on $A$. We have two base cases. Let us first consider when $A=base$. In the left-hand side of the equation, $\mono{base}\subst{\mono{B}}{\alpha} = base\subst{\mono{B}}{\alpha} = base$. In the right-hand side, $\mono{base\subst{B}{\alpha}}=base$.
 
	For the other base case that $A$ is a type variable, say, $\beta$, we do a case analysis on if $\beta$ is the same as $\alpha$ or not. When $\beta\not=\alpha$, this can be proved similarly as for the first base case. When $\beta=\alpha$, $\mono{\alpha}\subst{\mono{B}}{\alpha} = \alpha\subst{\mono{B}}{\alpha} = \mono{B} = \mono{\alpha\subst{B}{\alpha}}$.
	
	There are three inductive cases where $A$ is a function type, a polymorphic type, and a polymorphic location. 
	
	{\bf Case $A$ is $A_1\funL{a}A_2$:}

\begin{tabular}{l l l}
 		& $\mono{A_1\funL{a}A_2}\subst{\mono{B}}{\alpha}$ & by def. of $\mono{-}$ \\
 $=$ & $(\mono{A_1}\funL{a}\mono{A_2})\subst{\mono{B}}{\alpha}$ & \\
 $=$ & $\mono{A_1}\subst{\mono{B}}{\alpha}\funL{a}\mono{A_2}\subst{\mono{B}}{\alpha}$ & by I.H. \\
 $=$ & $\mono{A_1\subst{B}{\alpha}}\funL{a}\mono{A_2\subst{B}{\alpha}}$ & by def. of $\mono{-}$ \\
 $=$ & $\mono{A_1\subst{B}{\alpha}\funL{a}A_2\subst{B}{\alpha}}$ &  \\ 
 $=$ & $\mono{(A_1\funL{a}A_2)\subst{B}{\alpha}}$ &  \\ 
 \\
\end{tabular}

	{\bf Case $A$ is $\forall\beta.A_0$:}

\ \ i) $\beta=\alpha$ 

\begin{tabular}{l l l}
		& $\mono{\forall\beta.A_0}\subst{\mono{B}}{\alpha}$ & by def. of $\mono{-}$ \\
 $=$	& $(\forall\beta.\mono{A_0})\subst{\mono{B}}{\alpha}$ & \\		
 $=$	& $\forall\beta.\mono{A_0}$ & by def. of $\mono{-}$ \\
 $=$	& $\mono{\forall\beta.A_0}$ & \\ 		
 $=$	& $\mono{(\forall\beta.A_0)\subst{B}{\alpha}}$ & \\ 		
 \\
\end{tabular}

\ \ ii) $\beta\not=\alpha$

\begin{tabular}{l l l}
		& $\mono{\forall\beta.A_0}\subst{\mono{B}}{\alpha}$ & by def. of $\mono{-}$ \\
 $=$	& $(\forall\beta.\mono{A_0})\subst{\mono{B}}{\alpha}$ & \\	
 $=$	& $\forall\beta.(\mono{A_0}\subst{\mono{B}}{\alpha})$ & by I.H. \\	 
 $=$	& $\forall\beta.(\mono{A_0\subst{B}{\alpha}})$ &  \\	  
 $=$	& $\mono{\forall\beta.(A_0\subst{B}{\alpha})}$ &  \\	   
 $=$	& $\mono{(\forall\beta.A_0)\subst{B}{\alpha}}$ &  \\	
 \\    
\end{tabular}

	{\bf Case $A$ is $\forall l.A_0$:}

\begin{tabular}{l l l}
     & $\mono{\forall l.A_0}\subst{\mono{B}}{\alpha}$  & by def. of $\mono{-}$ \\
 $=$ & $(\mono{A_0\subst{\client}{l}}\times\mono{A_0\subst{\server}{l}})\subst{\mono{B}}{\alpha}$ & \\
 $=$ & $(\mono{A_0\subst{\client}{l}}\times\mono{A_0\subst{\server}{l}})\subst{\mono{B}}{\alpha}$ & \\
 $=$ & $\mono{A_0\subst{\client}{l}}\subst{\mono{B}}{\alpha}\times\mono{A_0\subst{\server}{l}\subst{\mono{B}}{\alpha}}$ & by I.H. \\ 
 $=$ & $\mono{(A_0\subst{\client}{l})\subst{B}{\alpha}}\times\mono{(A_0\subst{\server}{l})
\subst{B}{\alpha}}$ &  \\  
 $=$ & $\mono{(A_0\subst{B}{\alpha})\subst{\client}{l}}\times\mono{(A_0
\subst{B}{\alpha})\subst{\server}{l}}$ &  by def. of $\mono{-}$ \\
 $=$ & $\mono{\forall l.(A_0\subst{B}{\alpha})}$ &  \\  
 $=$ & $\mono{(\forall l.A_0)\subst{B}{\alpha}}$ &   
\end{tabular}

 \end{proof}
 
\begin{lemma}[Location polymorphism] 
Suppose $\tyenv=\{l_1,\cdots, l_n\}\cup\tyenv_0$  such that $\tyenv_0$ has no location variables. If \ $\typing{\tyenv}{\Loc}{M}{A}$ then $(\typing{\tyenv_0}{\Loc}{M}{A})\{a_1/l_1,\cdots, a_n/l_n\}$ for any $a_1,\cdots, a_n$ with the same height.
%\label{lemma:locationpolymorphism}
\end{lemma}
 \begin{proof} The proof can be done straightforwardly by induction on the height of the derivation tree for the typing judgment in the condition.
 \end{proof}
 
\subsection{Proofs  of Lemmas in Section  3.2.3 Semantics correctness}

\setcounter{lemma}{7}

\begin{lemma}[Substitution under  monomorphization] $\mono{M}\subst{\mono{W}}{x} = \mono{M\subst{W}{x}}$. 
%\label{lemma:substitutionmono}
\end{lemma}
 \begin{proof} We prove this lemma by the structural induction on $M$. 
	In the base case, $M=y$. When $y=x$, $\mono{x}\subst{\mono{W}}{x} = x\subst{\mono{W}}{x} = \mono{W} = \mono{x\subst{W}{x}}$. 
	When $y\not=x$, $\mono{y}\subst{\mono{W}}{x} = y\subst{\mono{W}}{x} = y = \mono{y} = \mono{y\subst{W}{x}}$.
	
	For the inductive cases, the proof is done as follows. 

	{\bf Case $M$ is $\lamL{a}{y}{L}$:}

\ \ i) $y=x$ 

\begin{tabular}{l l l}
     & $\mono{\lamL{a}{y}{L}}\subst{\mono{W}}{x}$  & by def. of $\mono{-}$ \\
 $=$ & $(\lamL{a}{y}{\mono{L}})\subst{\mono{W}}{x}$ \\
 $=$ & $\lamL{a}{y}{\mono{L}}$ \\ 
 $=$ & $\mono{\lamL{a}{y}{L}}$ \\  
 $=$ & $\mono{(\lamL{a}{y}{L})\subst{W}{x}}$ \\   
\end{tabular}     	

\ \ ii) $y\not=x$ 

\begin{tabular}{l l l}
     & $\mono{\lamL{a}{y}{L}}\subst{\mono{W}}{x}$  & by def. of $\mono{-}$ \\
 $=$ & $(\lamL{a}{y}{\mono{L}})\subst{\mono{W}}{x}$ & \\
 $=$ & $\lamL{a}{y}{(\mono{L}\subst{\mono{W}}{x})}$ & by I.H. \\ 
 $=$ & $\lamL{a}{y}{\mono{L\subst{W}{x}}}$ & by def. of $\mono{-}$ \\  
 $=$ & $\mono{\lamL{a}{y}{(L\subst{W}{x})}}$ &  \\   
 $=$ & $\mono{(\lamL{a}{y}{L})\subst{W}{x}}$ &  \\ 
 \\  
\end{tabular}  

	{\bf Case $M$ is $L N$:}

\begin{tabular}{l l l}
     & $\mono{L N}\subst{\mono{W}}{x}$  & by def. of $\mono{-}$ \\
 $=$ & $(\mono{L} \mono{N})\subst{\mono{W}}{x}$  &  \\     
 $=$ & $(\mono{L}\subst{\mono{W}}{x}) (\mono{N}\subst{\mono{W}}{x})$  & by I.H. \\        
 $=$ & $(\mono{L\subst{W}{x}}) (\mono{N\subst{W}{x}})$  &  \\     
 $=$ & $\mono{(L\subst{W}{x}) (N\subst{W}{x}})$  &  \\          
 $=$ & $\mono{(L N)\subst{W}{x}}$  &  \\   
\\      
\end{tabular}  

	{\bf Case $M$ is $\Lambda l.V$:}

\begin{tabular}{l l l}
     & $\mono{\Lambda l.V}\subst{\mono{W}}{x}$  & by def. of $\mono{-}$ \\
 $=$ & $(\mono{V\subst{\client}{l}},\mono{V\subst{\server}{l}})\subst{\mono{W}}{x}$  &  \\
 $=$ & $(\mono{V\subst{\client}{l}}\subst{\mono{W}}{x},\mono{V\subst{\server}{l}}\subst{\mono{W}}{x})$  & by I.H. \\    
 $=$ & $(\mono{(V\subst{\client}{l})\subst{W}{x}},\mono{(V\subst{\server}{l})\subst{W}{x}}$  &  \\   
 $=$ & $(\mono{(V\subst{W}{x})\subst{\client}{l}},\mono{(V\subst{W}{x})\subst{\server}{l}}$  & by def. of $\mono{-}$ \\      
 $=$ & $(\mono{\Lambda l.(V\subst{W}{x})}$  &  \\        
 $=$ & $(\mono{(\Lambda l.V)\subst{W}{x}}$  &  \\        
     \\     
\end{tabular} 

	{\bf Case $M$ is $\Lambda \alpha.V$:}

\begin{tabular}{l l l}
     & $\mono{\Lambda \alpha.V}\subst{\mono{W}}{x}$  & by def. of $\mono{-}$ \\
 $=$ & $(\Lambda \alpha.\mono{V})\subst{\mono{W}}{x}$  &  \\          
 $=$ & $\Lambda \alpha.(\mono{V}\subst{\mono{W}}{x})$  & by I.H. \\     
 $=$ & $\Lambda \alpha.(\mono{V\subst{W}{x}})$  &  by def. of $\mono{-}$ \\    
 $=$ & $\mono{\Lambda \alpha.(V\subst{W}{x})}$  &   \\  
 $=$ & $\mono{(\Lambda \alpha.V)\subst{W}{x}}$  &   \\   
 \\
\end{tabular} 

	{\bf Case $M$ is $L[a]$:}

\begin{tabular}{l l l}
     & $\mono{L[a]}\subst{\mono{W}}{x}$  & by def. of $\mono{-}$ \\
 $=$ & $(\pi_i\mono{L})\subst{\mono{W}}{x}$  & $i=1$ if $a=\client$, \ \ $i=2$ if $a=\server$ \\     
 $=$ & $\pi_i(\mono{L}\subst{\mono{W}}{x})$  & by I.H. \\ 
 $=$ & $\pi_i(\mono{L\subst{W}{x}})$  & \\  
 $=$ & $\mono{L\subst{W}{x}[a]}$  & by def. of $\mono{-}$ \\   
 $=$ & $\mono{(L[a])\subst{W}{x}}$  & \\    
\\
\end{tabular} 

	{\bf Case $M$ is $L[B]$:}

\begin{tabular}{l l l}
     & $\mono{L[B]}\subst{\mono{W}}{x}$  & by def. of $\mono{-}$ \\
 $=$ & $\mono{L}[\mono{B}]\subst{\mono{W}}{x}$  &  \\     
 $=$ & $(\mono{L}\subst{\mono{W}}{x})[\mono{B}]$  & by I.H. \\ 
 $=$ & $\mono{L\subst{W}{x}}[\mono{B}]$  & by def. of $\mono{-}$  \\  
 $=$ & $\mono{L\subst{W}{x}[B]}$  &  \\   
 $=$ & $\mono{(L[B])\subst{W}{x}}$  &  \\  
\end{tabular} 

 \end{proof}
 
\begin{lemma}[Type substitution over term under  monomorphization] $\mono{V}\subst{\mono{B}}{\alpha} = \mono{V\subst{B}{\alpha}}$. 
%\label{lemma:typesubstitutiontermmono}
\end{lemma}
 \begin{proof} 
	We prove a slightly general lemma as $\mono{M}\subst{\mono{B}}{\alpha} = \mono{M\subst{B}{\alpha}}$. 
	We prove the general lemma by the structural induction on $M$.
	For the base case $M=x$, $\mono{x}\subst{\mono{B}}{\alpha} = x \subst{\mono{B}}{\alpha} = x = \mono{x} = \mono{x \subst{B}{\alpha} }$.
	
	For the inductive cases, the proof is done as follows.

	{\bf Case $M$ is $\lamL{a}{x}{L}$:}

\begin{tabular}{l l l}
     & $\mono{\lamL{a}{x}{L}}\subst{\mono{B}}{\alpha}$ & by def. of $\mono{-}$ \\
 $=$ & $(\lamL{a}{x}{\mono{L}})\subst{\mono{B}}{\alpha}$ &  \\
 $=$ & $\lamL{a}{x}{(\mono{L}\subst{\mono{B}}{\alpha})}$ & by I.H. \\ 
 $=$ & $\lamL{a}{x}{\mono{L\subst{B}{\alpha}}}$ &  \\  
 $=$ & $\mono{\lamL{a}{x}{(L\subst{B}{\alpha})}}$ & by def. of $\mono{-}$  \\   
 $=$ & $\mono{(\lamL{a}{x}{L})\subst{B}{\alpha}}$ &  \\  
     \\
\end{tabular} 

	{\bf Case $M$ is $\Lambda l.V$:}

\begin{tabular}{l l l}
     & $\mono{\Lambda l.V}\subst{\mono{B}}{\alpha}$ & by def. of $\mono{-}$ \\
 $=$ & $(\mono{V\subst{\client}{l}},\mono{V\subst{\server}{l}})\subst{\mono{B}}{\alpha}$ & \\
 $=$ & $(\mono{V\subst{\client}{l}}\subst{\mono{B}}{\alpha},\mono{V\subst{\server}{l}}\subst{\mono{B}}{\alpha})$ & by applying I.H. twice \\ 
 $=$ & $(\mono{(V\subst{\client}{l})\subst{B}{\alpha}},\mono{(V\subst{\server}{l})\subst{B}{\alpha}}$ &  \\ 
 $=$ & $(\mono{(V\subst{B}{\alpha})\subst{\client}{l}},\mono{(V\subst{B}{\alpha})\subst{\server}{l}}$ &  by def. of $\mono{-}$  \\  
 $=$ & $\mono{\Lambda l.(V\subst{B}{\alpha})}$ &   \\  
 $=$ & $\mono{(\Lambda l.V)\subst{B}{\alpha}}$ &   \\  
     \\
\end{tabular}

	{\bf Case $M$ is $L N$:}

\begin{tabular}{l l l}
     & $\mono{L N}\subst{\mono{B}}{\alpha}$ & by def. of $\mono{-}$ \\
 $=$ & $(\mono{L}\mono{N})\subst{\mono{B}}{\alpha}$ & \\     
 $=$ & $(\mono{L}\subst{\mono{B}}{\alpha} (\mono{N}\subst{\mono{B}}{\alpha})$ & by applying I.H. twice \\      
 $=$ & $\mono{L\subst{B}{\alpha}} \mono{N\subst{B}{\alpha}}$ & by def. of $\mono{-}$ \\       
 $=$ & $\mono{(L\subst{B}{\alpha}) (N\subst{B}{\alpha})}$ &  \\  
 $=$ & $\mono{(L N)\subst{B}{\alpha}}$ &  \\ 
     \\
\end{tabular}  

	{\bf Case $M$ is $L[a]$:}

\begin{tabular}{l l l}
     & $\mono{L[a]}\subst{\mono{B}}{\alpha}$ & by def. of $\mono{-}$ \ \ ($i=1$ if $a=\client$, $i=2$ if $a=\server$) \\
 $=$ & $(\pi_i\mono{L})\subst{\mono{B}}{\alpha}$ &     \\
 $=$ & $\pi_i(\mono{L}\subst{\mono{B}}{\alpha})$ &  by I.H.   \\ 
 $=$ & $\pi_i\mono{L\subst{B}{\alpha}}$ & by def. of $\mono{-}$ \\  
 $=$ & $\mono{(L\subst{B}{\alpha})[a]}$ &  \\   
 $=$ & $\mono{(L[a])\subst{B}{\alpha}}$ &  \\   
 \\ 
\end{tabular}  

	{\bf Case $M$ is $\Lambda\beta.V$:}

\ \ i) $\beta=\alpha$ 
	
\begin{tabular}{l l l}
     & $\mono{\Lambda\alpha.V}\subst{\mono{B}}{\alpha}$ & by def. of $\mono{-}$ \\
 $=$ & $(\Lambda\alpha.\mono{V})\subst{\mono{B}}{\alpha}$ &    \\  
 $=$ & $\Lambda\alpha.\mono{V}$ & by def. of $\mono{-}$ \\
 $=$ & $\mono{\Lambda\alpha.V}$ &  \\  
 $=$ & $\mono{(\Lambda\alpha.V)\subst{B}{\alpha}}$ &   \\
     \\
\end{tabular} 

\ \ i) $\beta\not=\alpha$ 
	
\begin{tabular}{l l l}
     & $\mono{\Lambda\beta.V}\subst{\mono{B}}{\alpha}$ & by def. of $\mono{-}$ \\
 $=$ & $(\Lambda\beta.\mono{V})\subst{\mono{B}}{\alpha}$ &    \\       
 $=$ & $\Lambda\beta.(\mono{V}\subst{\mono{B}}{\alpha})$ & by I.H.   \\        
 $=$ & $\Lambda\beta.(\mono{V\subst{B}{\alpha}})$ &  \\         
 $=$ & $\mono{\Lambda\beta.(V\subst{B}{\alpha})}$ &  \\     
 $=$ & $\mono{(\Lambda\beta.V)\subst{B}{\alpha}}$ &  \\   
     \\
\end{tabular} 

	{\bf Case $M$ is $L[A]$:}

\begin{tabular}{l l l}
     & $\mono{L[A]}\subst{\mono{B}}{\alpha}$ & by def. of $\mono{-}$ \\
 $=$ & $(\mono{L}[\mono{A}])\subst{\mono{B}}{\alpha}$ & \\     
 $=$ & $(\mono{L}\subst{\mono{B}}{\alpha})[\mono{A}\subst{\mono{B}}{\alpha}]$ & by I.H. and Lemma \ref{lemma:typesubstitutionmono} \\  
 $=$ & $\mono{L\subst{B}{\alpha}}[\mono{A\subst{B}{\alpha}}]$ & by def. of $\mono{-}$ \\ 
 $=$ & $\mono{L\subst{B}{\alpha}[A\subst{B}{\alpha}]}$ & \\ 
 $=$ & $\mono{(L[A])\subst{B}{\alpha}}$ & \\  
     \\
\end{tabular} 

 \end{proof}

%%%    3.3 A polymorphic RPC calculus

\end{appendix}

\end{document}